\documentclass[10pt,twocolumn,twoside]{IEEEtran}
\usepackage{amsmath}
\usepackage{amssymb}
\usepackage{amsfonts}
\usepackage{mathrsfs}
\usepackage{graphicx}
\usepackage{epstopdf}
\usepackage[noadjust]{cite} 
\usepackage{subfig}
\usepackage{amsthm}
\theoremstyle{plain}
\newtheorem{theorem}{Theorem}
\theoremstyle{definition}
\newtheorem{Remarks}{Remark}
\theoremstyle{plain}
\newtheorem{lemma}{Lemma}
\theoremstyle{plain}

\usepackage{paralist}
\begin{document}

\addtolength{\abovedisplayskip}{-3.2pt}
\addtolength{\belowdisplayskip}{-3.2pt}
\addtolength{\abovedisplayshortskip}{-2pt}
\addtolength{\belowdisplayshortskip}{-2pt}
\addtolength{\belowcaptionskip}{-5pt}


\title{Spectrum Sensing using Distributed Sequential Detection via Noisy Reporting MAC}
\author{Jithin~K.~Sreedharan and Vinod~Sharma,~\IEEEmembership{Senior~Member,~IEEE} 
\thanks{Jithin K.\ Sreedharan is with INRIA Sophia Antipolis, France. Email: jithin.sreedharan@inria.fr}
\thanks{Vinod Sharma is with Department of Electrical Communication Engineering, Indian Institute of Science, Bangalore, India. Email: vinod@ece.iisc.ernet.in}
\thanks{Preliminary versions of this paper have been presented in NCC 2011, WCNC 2011 and Allerton 2011.}}
\maketitle

\begin{abstract}
This paper considers cooperative spectrum sensing algorithms for Cognitive Radios which focus on reducing the number of samples to make a reliable detection. We develop an energy efficient detector with low detection delay using decentralized sequential hypothesis testing. Our algorithm at the Cognitive Radios employs an asynchronous transmission scheme which takes into account the noise at the fusion center. We start with a distributed algorithm, DualSPRT, in which Cognitive Radios sequentially collect the observations, make local decisions using SPRT (Sequential Probability Ratio Test) and send them to the fusion center. The fusion center sequentially processes these received local decisions corrupted by noise, using an SPRT-like procedure to arrive at a final decision. We theoretically analyse its probability of error and average detection delay. We also asymptotically study its performance. Even though DualSPRT performs asymptotically well, a modification at the fusion node provides more control over the design of the algorithm parameters which then performs better at the usual operating probabilities of error in Cognitive Radio systems. We also analyse the modified algorithm theoretically. Later we modify these algorithms to handle uncertainties in SNR and fading.
\end{abstract}


\section{Introduction}
Presently there is a scarcity of spectrum due to the proliferation of wireless services. Cognitive Radios (CRs) are proposed as a solution to this problem. They access the spectrum licensed to existing communication services (primary users) opportunistically and dynamically without causing much interference to the primary users. This is made possible via spectrum sensing by the Cognitive Radios (secondary users), to gain knowledge about the spectrum usage by the primary devices. However due to the strict spectrum sensing requirements (\cite{Akyildiz_PC2011}) and the various inherent wireless channel impairments spectrum sensing has become one of the main challenges faced by the Cognitive Radios.

Multipath fading, shadowing and hidden node problem cause serious problems in spectrum sensing. Cooperative (decentralized or distributed) spectrum sensing in which different cognitive radios interact with each other exploiting spatial diversity, (\cite{Unnikrishnan_JSTSP2008,Akyildiz_PC2011}) is proposed as an answer to these problems. Also it reduces the probability of false alarm and the probability of miss-detection. Cooperative spectrum sensing can be either centralized or distributed (\cite{Akyildiz_PC2011}). In the centralized algorithm a central unit gathers sensing data from the Cognitive Radios and identifies the spectrum usage (\cite{Quan_SPM2008}). On the other hand, in the distributed case each secondary user (SU) collects observations, makes a local decision and sends to a fusion center (FC) to make the final decision. Centralized algorithms provide better performance but also have more communication overhead in transmitting all the data to the fusion node. In the distributed case, the information that is exchanged between the secondary users and the fusion node can be a soft decision (summary statistic) or a hard decision. Soft decisions can give better gains at the fusion center but also consume higher bandwidth at the control channels (used for sharing information among secondary users). However hard decisions provide as good a performance as soft decisions when the number of cooperative users increases (\cite{Quan_SPM2008}).

Spectrum sensing problem can be formulated in different ways, two of them being Neyman-Pearson framework (fixed sample size detection) and sequential detection framework which reduces the number of samples taken for deciding if a primary is transmitting or not. Sequential framework enables decision more quickly than the fixed sample size counterpart (\cite{Siegmund_SATC_Book}). Also, there are two types of sequential detection: one can consider detecting when a primary turns ON (or OFF) (change detection, see \cite{Li_WCOM2010,Banerjee_WCOM} and the references therein) or just testing the hypothesis whether the primary is ON or OFF (\cite{Zou_TSP,Yilmaz_Asilomar2012,Shei_PIMRC2008} and references therein). In \cite{Li_WCOM2010}, cooperative spectrum sensing under sequential change detection framework with no coordination between the secondary users is considered, and random broadcast policies and several improvements are proposed. In sequential hypothesis testing one considers the case where the status of the primary channel is known to change very slowly, e.g., detecting occupancy of a TV transmission. Usage of idle TV bands by the Cognitive network is being targeted as the first application for cognitive radio. 
In this setup (minimising the expected sensing time with constraints on probability of errors) Walds' SPRT (Sequential Probability Ratio Test) provides the optimal performance for a single Cognitive Radio (\cite{Siegmund_SATC_Book}). But the optimal solutions for cooperative setup are not available (\cite{Veeravalli1999}).

In this paper, we consider sequential hypothesis testing in cooperative setup. We first propose a decentralized algorithm DualSPRT, in which the secondary users sequentially collect the observations, make local decisions using SPRT and send them to the fusion center. Then the fusion center sequentially processes these received local decisions corrupted by noise, using a new sequential test, to arrive at a final decision. Unlike some of the previous works on cooperative spectrum sensing using sequential testing (see \cite{Yilmaz_Allerton2012_journal,Shei_PIMRC2008} and references therein) we analyse this algorithm theoretically also. Feedback from the fusion node to the CRs can possibly improve the performance. However that also requires an extra signalling channel which may not be available and has its own cost. Thus in our framework we assume that there is no feedback from the fusion center to the CRs. Furthermore, we consider the receiver noise at the fusion node and use physical layer fusion to reduce the transmission time of the decisions by the local nodes to the fusion node. 

In sequential decentralized detection framework, optimization needs to be performed jointly over sensors and fusion center policies as well as over time. Unfortunately, this problem is intractable for most of the sensor configurations (\cite{Mei_TIT2008,Veeravalli1999}). Specifically there is no optimal solution available for sensor configurations with no feedback from fusion center and limited local memory, which is more relevant in practical situations. Recently \cite{Fellouris_TIT2011} and \cite{Mei_TIT2008} proposed asymptotically optimal (order 1 (Bayes) and order 2 respectively) decentralized sequential hypothesis tests for such systems with full local memory. But these models do not consider noise at the fusion center and assume a perfect communication channel between the CR nodes and the fusion center. Also, often asymptotically optimal tests do not perform well at a finite number of observations. 

Noisy channels between local nodes and fusion center are considered in \cite{Yilmaz_Allerton2012_journal} in decentralized sequential detection framework. But optimality of the tests are not discussed and the paper is more focussed on finding the best signalling schemes at the local nodes with the assumption of parallel channels between local nodes and the fusion center. Also fusion center tests are based on the assumption of perfect knowledge of local node probability of false alarm and probability of miss-detection. 

We study asymptotic performance of DualSPRT, with fusion center noise. It can approach the optimal centralized sequential solution (in Bayes and frequentist sense), which does not consider noise at FC. We assume a MAC (Multiple Access Channel) as the reporting channel at the fusion center and the test is not based on the local node probability of error. Later we modify DualSPRT to improve its performance. The parameters of the modified algorithm are easier to fine tune also. Furthermore we introduce a new way of quantizing SPRT decisions of local nodes and extend this algorithm to cover SNR uncertainties and fading channels. We also study its performance theoretically. We have seen via simulations that our algorithm works better than the algorithm in \cite{Mei_TIT2008} and almost as well as the algorithm in \cite{Fellouris_TIT2011} even when the fusion center noise is not considered and MAC layer transmission delays are ignored in \cite{Fellouris_TIT2011} and \cite{Mei_TIT2008}.

In addition, we generalize our algorithm to include uncertainty in the received Signal to Noise Ratio (SNR) at the CRs and fading channels between primary and CR. This requires a composite hypothesis testing extension to the decentralized sequential detection problem and is not considered in any of the above references. \cite{Zou_TSP,Yilmaz_Asilomar2012} also proposed cooperative sequential algorithms for spectrum sensing, but neither of them deal with the fusion center noise and SNR uncertainty case. 

This paper is organised as follows. Section \ref{sec:DualSPRT_sys_mod} presents the model. Section \ref{sec:DualSPRT} provides the DualSPRT algorithm. An approximate theoretical performance of the algorithm is also provided. Section \ref{sec:DualSPRT_asy_opt} studies the asymptotic performance of DualSPRT. In Section \ref{sec:SPRT-CSPRT} we improve over DualSPRT. We compare the different versions so obtained and also compare them with existing asymptotically optimal decentralized sequential algorithms. Sections \ref{sec:SNR_uncertainty_fading} extends these algorithms to consider the effect of fading and SNR uncertainty. Section \ref{sec:conclusions} concludes the paper.

\section{System Model}
\label{sec:DualSPRT_sys_mod}
We consider a Cognitive Radio system with one primary transmitter and $L$ secondary users. The $L$ nodes sense the channel to detect the spectral holes. The decisions made by the secondary users are transmitted to a fusion node via a reporting MAC for it to make a final decision.

Let $X_{k,l}$ be the observation made at secondary user $l$ at time $k$. The $\{X_{k,l},~k\geq 1\}$ are independent and identically distributed (i.i.d.). It is assumed that the observations are independent across Cognitive Radios. Based on $\{X_{n,l},~n\leq k\}$ the secondary user $l$ transmits $Y_{k,l}$ to the fusion node. It is assumed that the secondary nodes are synchronised so that the fusion node receives $Y_k = \sum_{l=1}^L Y_{k,l} + Z_{k}$, where $\{Z_{k}\}$ is i.i.d. receiver noise. The fusion center uses $\{Y_k\}$ and makes a decision. The observations $\{X_{k,l}\}$ depend on whether the primary is transmitting (Hypothesis $H_1$) or not (Hypothesis $H_0$) as
{\allowdisplaybreaks
\begin{eqnarray}
\text{Under }H_0:& &\,X_{k,l}=\zeta_{k,l},  ~~~k=1,2,\ldots,\nonumber \\
\text{Under }H_1:& &\,X_{k,l}=h_lS_k+\zeta_{k,l},  ~~~k=1,2,\ldots,\nonumber
\end{eqnarray}}
where $h_l$ is the channel gain of the $l^{th}$ user, $S_k$ is the primary signal and $\zeta_{k,l}$ is the observation noise at the ${l^{th}}$ user at time $k$. We assume $\{\zeta_{k,l},k\geq 1\}$ are i.i.d. Let $N$ be the time to decide on the hypothesis by the fusion node. We assume that $N$ is much less than the coherence time of the channel so that the slow fading assumption is valid. This means that ${h_l}$ is random but remains constant during the spectrum sensing duration. 

The general problem is to develop a distributed algorithm in the above setup which solves the problem:
\begin{align}
\min E_{DD} &\stackrel{\Delta}{=}E_i[N]\nonumber \\
\label{EDD_eq}
\text{subject to } P_1({\text{Reject $H_1$}}) &\leq  \alpha_1 \,\& \, P_0({\text{Reject $H_0$}}) \leq \alpha_0,
\end{align}
where $P_i$ is the probability measure and $E_i$ the expectation when $H_i$ is the true hypothesis, $i\in \{0,1\}$, and $0\leq \alpha_0, \alpha_1 \leq 1$. We will separately consider $E_1[N]$ and $E_0[N]$. It is well known that for a single node case ($L=1$) Wald's SPRT performs optimally in terms of reducing $E_1[N]$ and $E_0[N]$ for given probability of errors. Motivated by the optimality of SPRT for a single node (and DualCUSUM in \cite{Banerjee_WCOM}), we propose using DualSPRT in the next section and study its performance. 
 
We use $P_{MD}$ for $P_1$(reject $H_1$) and $P_{FA}$ for $P_0$(reject $H_0$). In case of $E_{DD}$, hypothesis under consideration can be understood from the context. 

\section{Decentralized Sequential Tests: DualSPRT}
\label{sec:DualSPRT}
In this section we develop DualSPRT algorithm for decentralized sequential detection and also study its performance. 

\subsection{DualSPRT algorithm}
\label{sec:DualSPRT_dualsprt}
To explain the setup and analysis we start with the simple case, where the channel gain, $h_l=1$ for all $l's$. We will consider fading in the next section. DualSPRT is as follows:
\begin{enumerate}
\item Secondary node $l$, computes at step $k$, 
\begin{IEEEeqnarray}{rCl}
\label{SensorSPRT}
W_{0,l} &=& 0, \nonumber \\
W_{k,l} &=& W_{k-1, l} + \log \left[f_{1,l}\left(X_{k,l}\right)\left / f_{0,l}\left(X_{k,l} \right. \right) \right] ,k\geq1,\nonumber \IEEEeqnarraynumspace
\end{IEEEeqnarray}
where $f_{1,l}$ is the density of $X_{k,l}$ under ${H_1}$ and $ f_{0,l}$ is the density of $X_{k,l}$ under ${H_0}$ (w.r.t.\ a common distribution).
\item Secondary node $l$ transmits a constant $b_1$ at time $k$ if $W_{k,l}\geq \gamma_{1,l} $ or transmits $b_0$ when $W_{k,l}\leq -\gamma_{0,l}$. When $W_{k,l}$ does not cross the interval $(-\gamma_{0,l}, \gamma_{1,l})$, node $l$ does not transmit anything, i.e.,
\[Y_{k,l}=b_1\, \mathbb{I}\{W_{k,l}\geq \gamma_{1,l}\}+b_0 \, \mathbb{I}\{W_{k,l}\leq -\gamma_{0,l}\}\]
where $\gamma_{0,l},\gamma_{1,l}>0$ and $\mathbb{I}\{A\}$ denotes the indicator function of set A. Parameters $b_1, b_0, \gamma_{1,l}, \gamma_{0,l}$ are chosen appropriately.
\item Physical layer fusion is used at the fusion Centre, i.e.,
$Y_k = \sum_{l=1}^L Y_{k,l} + Z_{k}$, where $\{Z_{k}\}$ is the i.i.d.\ noise at the fusion node.
\item Finally, fusion center calculates the log-likelhood ratio:
\begin{IEEEeqnarray}{rCl}
\label{eq:ch3:FuseCUSUM}
 F_k &=&F_{k-1} + \log \left[g_{\mu_1}\left(Y_{k}\right)\left / g_{-\mu_0}\left(Y_{k} \right. \right) \right],\\
 F_0 &=& 0,\,\mu_0>0,\,\mu_1>0,\nonumber
\end{IEEEeqnarray}
where $g_{-\mu_0}$ is the density of $Z_{k}-\mu_0$ and $g_{\mu_1}$ is the density
of $Z_{k}+\mu_1$, $\mu_0$ and $\mu_1$ being positive constants appropriately chosen. 
\item The fusion center decides about the hypothesis at time $N$ 
where\[N=  \inf \{ k: F_k \geq \beta_1 \text{ or } F_k \leq -\beta_0\}\] and $\beta_0,\beta_1>0$. The decision at time $N$ is $H_1$ if $F_N \geq \beta_1$, otherwise $H_0$.

\end{enumerate}

Performance of this algorithm depends on ($\gamma_{1,l},\gamma_{0,l},\beta_1,\beta_0,b_1,b_0,\mu_1,\mu_0$). In particular these parameters should be chosen such that the overall probabilities of error are less than $\alpha_1$ and $\alpha_0$ respectively. Any prior information available about $H_0$ or $H_1$ can be used to decide constants (via, say, formulating this problem in the Bayesian framework; we will comment on this again). Also we choose these parameters such that the probability of false alarm/miss-detection, $P_{fa}/P_{md}$ at local nodes is higher than $P_{FA}/P_{MD}$. A good set of parameters for given SNR values can be obtained from our analysis below.

Deciding at local nodes and transmitting decisions to the fusion node reduces the transmission rate and transmit energy used by the local nodes in communication with the fusion node. Also, physical layer fusion in Step 3 reduces transmission time, but requires synchronisation of different local nodes. If synchronisation is not possible, then some other MAC algorithm, e.g., TDMA can be used with channel coding. But this will incur extra delay. 

Using sequential tests at SUs and at FC (without physical layer synchronization and fusion receiver noise) has been shown to perform well in (\cite{Zou_TSP}, \cite{Shei_PIMRC2008}). In the next subsection we analyse the performance under our setup.

\subsection{Performance Analysis}
\label{sec:DualSPRT_per_ana}
We first provide the analysis for the mean detection delay $E_{DD}$ and then for $P_{FA}$. 

KL-divergence of two probability distributions $P$ and $Q$ on the same measurable space $(\Omega,\mathcal{F})$ is defined as
\begin{equation}
\label{eq:ch3:KL-divergence}
D(P||Q)=\left\{
\begin{array}{ll}
\int \log \frac{dP}{dQ} \, dP &, \text{ if } P<<Q,\\
\infty &, \text{ otherwise ,}
\end{array}
\right.
\end{equation}
where $P<<Q$ denotes that $P$ is absolutely continuous w.r.t.\ $Q$. More explicitly, at node $l$, let
\[\delta_{i,l}=E_{i}\left[\log \frac{f_{1,l}(X_{k,l})}{f_{0,l}(X_{k,l})} \right],\, \rho^2_{i,l}=Var_{H_i} \left[\log \frac{f_{1,l}(X_{k,l})}{f_{0,l}(X_{k,l})} \right]. \]
Then $\delta_{1,l}=D(f_{1,l}||f_{0,l})$ and $\delta_{0,l}=-D(f_{0,l}||f_{1,l})$. We will assume $\delta_{i,l}$ finite throughout this paper. Sometimes we will also need $\rho^2_{i,l} < \infty$. When the true hypothesis is $H_1$, by Jensen's Inequality, $\delta_{1,l} >0$ and when it is $H_0$, $\delta_{0,l}<0$. At secondary node $l$, SPRT sum $\{W_{k,l},k\geq 0\}$ is a random walk with drift given by $\delta_{i,l}$ under the true hypothesis $H_i$. 

Let
\[N_{l}=\inf \{k: W_{k,l} \notin (-\gamma_{0,l},\gamma_{1,l})\}, \] \[N_{l}^1=\inf \{k: W_{k,l} \geq \gamma_{1,l} \} \text{ and }
N_{l}^0=\inf \{k: W_{k,l} \leq -\gamma_{0,l} \}. \]
Then $N_l=\min\{N_l^0,N_l^1\}$. Also let $N^0=\inf\{k:F_k\leq -\beta_0\}$ and $N^1=\inf\{k:F_k\geq \beta_1\}$. Then stopping time of DualSPRT, $N=\min(N^1,N^0)$.

For simplicity in the rest of this section, we take $\gamma_{1,l}=\gamma_{0,l}=\gamma$, $\beta_1=\beta_0=\beta$, $b_1=-b_0=b$ and $\mu_1=\mu_0=\mu$. Of course the analysis will carry over for the general case.

For convenience we summarize the important notation used in this paper in Table \ref{tab:ch3:list_symbols}. Notation specific to some algorithms are also mentioned separately.
{\renewcommand{\arraystretch}{1.2}
\begin{table}[!tbh]
\begin{center}
\begin{tabular}[t]{|c||l|}\hline
  Notation & Meaning\\ \hline \hline
    $L$ & Number of CRs \\ \hline
  $X_{k,l}$ & Observation at CR $l$ at time $k$  \\ \hline
  $Y_{k,l}$ & Transmitted value from CR $l$ to FC at time $k$.\\ \hline
  $Y_k$ & FC observation at time $k$ \\ \hline
  $h_l$ & Channel gain of the $l^{\text{th}}$ CR \\ \hline
  $\zeta_{k,l}$ & Observation noise at CR $l$ at time $k$ \\ \hline
  $Z_k$ & FC MAC noise at time $k$ \\ \hline
  $f_{i,l}$, $g_{\mu}$ & PDF of $X_{1,l}$ under $H_i$, PDF of $Z_k+\mu$\\ \hline
  $W_{k,l}$ & Test statistic at CR $l$ at time $k$ \\ \hline
  $F_k$ & Test statistic at FC at time $k$ $(1)$\\ \hline
  $F_k^1$, $F_k^0$ & Test statistics at FC $(2)$\\ \hline
  $\xi_k$ & LLR at FC (1)\\ \hline
  $\xi_k^*$ & LLR when all CR's transmit wrong decisions (1)\\ \hline
  $\theta_i$ & Worst case value of $E_i[\xi_k],=E_i[\xi_k^*]$ (1) \\ \hline
  $F_n^*$, $\widehat{F}_n^*$ & $\sum_{k=1}^{n}\xi_k^*$, $\sum_{k=1}^{n}|\xi_k^*|$ (1)\\ \hline
  $\mathcal{A}^i$, $\Delta(\mathcal{A}^i)$ & \{all CRs transmit $b_i$ under $H_i$ \}, $E_i[\xi_k|\mathcal{A}^i]$ (1)
\\ \hline  
  $\gamma_{1,l}, \gamma_{0,l}$ & Thresholds at CR  $l$ (1,2)\\ \hline
  $g(t)$ & Threshold at CR $(3,4)$ \\ \hline
  $\beta_1, \beta_0$ & Thresholds at FC \\ \hline
  $\mu_1,\mu_0$ & Design parameters in FC LLR \\ \hline
  $b_1,b_0$ & Transmitting values to the FC at CR $({1,3})$ \\ \hline
  $b_j^i$ & Transmitting values to the FC at CR $(2,4)$\\ \hline
  $N$ & First time $F_k$ crosses $(-\beta_0, \beta_1)$ (1)\\ \hline
  $N^1, N^0$ & First time $F_k$ crosses $\beta_1$, crosses $-\beta_0$ (1)\\ \hline
  $N_l,N_l^1,N_l^0$ & Corresponding values of $N$, $N^1$, $N^0$ at CR $l$ (1)\\ \hline
  $N_l(g,c)$ & First time $W_{n,l}$ crosses $g(nc)$ at CR $l$ $(3,4)$\\ \hline
  $\tau_\beta$, $T_\beta$ & First time $F_k^0$ crosses $-\beta_0$ , $F_k^1$ crosses $\beta_1$ (2)\\ \hline
  $\delta_{i,l},\rho^2_{i,l}$ & Mean and variance of LLR at CR $l$ under $H_i$ \\ \hline
  $\delta_{i,FC}^j$ & Mean of LLR at FC under $H_i$ when $j$ CRs transmit \\ \hline
  $t_j$ & Time epoch when $\delta_{i,FC}^{j-1}$ changes to $\delta_{i,FC}^j$ $(1)$\\ \hline
  $T_j$ & Time epoch when $\delta_{i,FC}^{j-1}$ changes to $\delta_{i,FC}^j$ $(2)$\\ \hline
  $\bar{F_j}$ & $E[F_{t_j-1}]$\\ \hline
  $D_{tot}^0$, $D_{tot}^1$ & $\sum_{l=1}^L D(f_{0,l}||f_{1,l})$, $\sum_{l=1}^L D(f_{1,l}||f_{0,l})$\\ \hline
  $r_l$, $\rho_l$ & $D(f_{0,l}||f_{1,l})/D_{tot}^0$, $D(f_{1,l}||f_{0,l})/D_{tot}^1$\\ \hline
  $\tau_l(c)$ & Last time RW with drift $\delta_{0,l}$ will be above $-|\log c|$\\ \hline
  $\tau(c)$ & $\displaystyle \max_{1\leq l \leq L} \tau_l(c)$ \\ \hline
  $R_i$ & $\displaystyle\min_{1\leq l \leq L} -\log \inf_{t \geq 0} E_i\Big[\exp \left({-t \log \frac{f_{1,l(X_{1,l})}}{f_{0,l}(X_{1,l})}}\right)\Big]$ \\ \hline
  $G_i$, $\widehat{G}_i$, $g_i$, $\widehat{g}_i$ & CDF of $|\xi_1^*|$, $\xi_1^*$, MGF of $|\xi_1^*|$, $\xi_1^*$\\ \hline
  $\Lambda_i(\alpha)$, $\widehat{\Lambda}_i(\alpha)$ & $\sup_{\lambda}(\alpha \lambda-\log g_i(\lambda))$, $\sup_{\lambda}(\alpha \lambda-\log \widehat{g}_i(\lambda))$ \\ \hline
  $\alpha_i^+$ & $\text{ess }\sup |\xi_1^*|$ \\ \hline
  $\mathcal{R}_c(\delta)$ & Bayes Risk of test $\delta$ with cost $c$\\ \hline
  & First time RW \\
  $\nu(a)$ & $\{\log \frac{g_{\mu_1(Z_k)}}{g_{-\mu_0(Z_k)}}+(\Delta(\mathcal{A}^0)-E_0[\log \frac{g_{\mu_1(Z_k)}}{g_{-\mu_0(Z_k)}}])$  \\
  & $ k \geq \tau(c)+1\}$ crosses $a$.\\ \hline

\end{tabular}\caption{List of important notations. 1- DualSPRT, 2- SPRT-CSPRT, 3- GLR-SPRT, 4- GLR-CSPRT}
\label{tab:ch3:list_symbols}
\end{center}
\end{table}}

\subsubsection{$E_{DD}$ Analysis}
\label{subsub:ch3:edd_analysis}
At the fusion node $F_k$ crosses $\beta$ under ${H_1}$ when a sufficient number of local nodes transmit $b_1$. The dominant event occurs when the number of local nodes transmitting are such that the mean value of the increments of the sum ${F_k}$ will just have turned positive. In the following we find the mean time to this event and then the time to cross $\beta$ after this. The $E_{DD}$ analysis is same under hypothesis $H_0$ and $H_1$. Hence we provide the analysis for $H_1$.

The following lemmas provide justification for considering only the events $\{N_l^i\}$ and $\{N^i\}$ for analysis of $E_{DD}=E_i[N]$.
\begin{lemma}
\label{lemma:ch3:N_l_N_wp1}
For $i=0,1$, $P_{i}(N_l=N_l^i) \to 1$ as $\gamma \to \infty$ and $P_{i}(N=N^i) \to 1$ as $\gamma \to \infty$ and $\beta \to \infty$.
\end{lemma}
\begin{IEEEproof}
From random walk results (\cite[Chapter IV]{GUT_BOOK_2009}) we know that if a random walk has negative drift then its maximum is finite with probability one. This implies that $P_{i}(N_l^j<\infty) \to 0$ as $\gamma \to \infty$ for $i \neq j$ but $P_{i}(N_l^i<\infty)=1$ for any $\gamma < \infty$. Thus $P_{i}(N_l=N_l^i) \to 1$ as $\gamma \to \infty$. This also implies that as $\gamma \to \infty$, the mean of increments of $F_k$ is positive for $H_1$ and negative for $H_0$. Therefore, $P_{i}(N=N^i) \to 1$ as $\gamma \to \infty$ and $\beta \to \infty$.
\end{IEEEproof}

\begin{lemma}
\label{lemma:ch3:N_Ni_exp_conv}
Under $H_i$, $i=0,1$ and $j \neq i$,
\begin{enumerate}
\item[(a)] $|N_l-N_l^i| \to 0$ a.s. as $\gamma \to \infty$ and $\lim_{\gamma \to \infty} \frac{N_l}{\gamma}=\lim_{\gamma \to \infty} \frac{N_l^i}{\gamma}=\frac{1}{D(f_{i,l}||f_{j,l})}$ a.s. and in $L^1$.
\item[(b)] $|N-N^i| \to 0$ a.s. and $\lim \frac{N}{\beta}=\lim \frac{N^i}{\beta}$ a.s. and in $L^1$, as $\gamma \to \infty $ and $\beta$ $\to \infty$.
\end{enumerate}   
\end{lemma}

\begin{IEEEproof}
Under $H_0$,
\begin{equation}
\label{eq:ch3:lemma2_proof_bound}
N_l^0\,\mathbb{I}\{N_l^0<N_l^1\} \leq N_l \leq N_l^0,
\end{equation}
and since $P_0[N_l^0 < N_l^1] \to 1$ as $\gamma \to \infty$, $|N_l^0-N_l| \to 0$ a.s. as $\gamma \to \infty$. Also from Random Walk results (\cite[p.~88]{GUT_BOOK_2009}), $N_l^0/\gamma \to 1/D(f_{0,l}||f_{1,l})$ a.s. and $E[N_l^0]/\gamma \to 1/D(f_{0,l}||f_{1,l})$. Thus we also obtain $N_l/\gamma \to 1/D(f_{0,l}||f_{1,l})$ a.s. and in $L^1$. Similarly the corresponding results hold for $N$ and $N^0$ as $\gamma$ and $\beta$ $\to \infty$. (\ref{eq:ch3:lemma2_proof_bound}) holds in the expected sense also. 
\end{IEEEproof}

Thus when $\gamma$ is large, we can approximate $E_1[N_l]$ by $\gamma D(f_{1,l}||f_{0,l})$. Also under $H_1$, by central limit theorem for the first passage time $N_l^1$ (Theorem 5.1, Chapter III in \cite{GUT_BOOK_2009}),
\begin{equation}
\label{eq:ch3:N_l_1_pdf}
N_{l}^1\sim\mathcal{N}(\frac{\gamma}{\delta_{1,l}}, \frac{\rho^2_{1,l}\, \gamma}{\delta_{1,l}^3}),
\end{equation}
where $\mathcal{N}(a,b)$ denotes Gaussian distribution with mean $a$ and variance $b$. From Lemma \ref{lemma:ch3:N_Ni_exp_conv}, we can use this result for $N_l$ also. Similarly we can obtain the results under $H_0$ and at the fusion node. Let $\delta_{i,FC}^j$ be the mean of increments of the fusion center test sum $F_k$, under $H_i$, when $j$ local nodes are transmitting. Let $t_j$ be the point at which the mean of increments of $F_k$ changes from $\delta_{i,FC}^{j-1}$ to $\delta_{i,FC}^j$ and let $\bar{F_j}=E[F_{t_j-1}]$, the mean value of $F_k$ just before transition epoch $t_j$. The following lemma holds.

\begin{lemma}
\label{lemma:ch3:VS_argument_reg_finiteness_st}
Under $H_i$, $i=0,1$, as $\gamma \to \infty$,\\
$P_{i}(\text{Decision at time $t_k$ is $H_i$ and}$\\ \hspace*{0.6 cm}$\text{$t_k$ is the $k^{th}$ order statistics of $N_1^i,N_2^i,\ldots,N_L^i$}) \to 1$.
\end{lemma}

\begin{proof}
From Lemma \ref{lemma:ch3:N_l_N_wp1}, \\
$P_{i}(\text{Decision at time $t_k$ is $H_i$ and}$\\ \hspace*{0.6 cm} $\text{$t_k$ is the $k^{th}$ order statistics of $N_1^i,N_2^i,\ldots,N_L^i$})$ \\
\hspace*{0.5 cm}$\geq P_{i}(N_l^i<N_l^j, j \neq i, l=1,\ldots,L) \to 1,\, \text{ as } \gamma \to \infty.$
\end{proof}

We use Lemma \ref{lemma:ch3:N_l_N_wp1}-\ref{lemma:ch3:VS_argument_reg_finiteness_st} and equation (\ref{eq:ch3:N_l_1_pdf}) in the following to obtain an approximation for $E_{DD}$ when $\gamma$ and $\beta$ are large. Large $\gamma$ and $\beta$ are needed for small probability of error. Then we can assume that the local nodes are making correct decisions. Although $F_k$ is a random walk before $t_1$, it is not so between  $t_j$ and $t_{j+1}$ for $j=1,\ldots,L$. But we assume that in the following approximation.

Let
\[l^*=min\{j:\delta_{1,FC}^j>0 \text{ and } \frac{\beta-\bar{F_j}}{\delta_{1,FC}^j}< E[t_{j+1}]-E[t_j]\}.\]
\noindent $\bar{F_j}$ can be iteratively calculated as
\begin{equation}
\label{eq:ch3:bar{F_j}}
\bar{F_j}=\bar{F}_{j-1}+\delta_{1,FC}^j\,(E[t_j]-E[t_{j-1}]),\, \bar{F}_0=0.
\end{equation}
Note that $\delta_{1,FC}^j\, (0 \leq j \leq L)$ can be found by assuming $E_1[Y_k]$ as $bj$ and $t_j$ as the $j^{th}$ order statistics of $\{N_l^1,\, 0\leq l \leq L \}$. The Gaussian approximation (\ref{eq:ch3:N_l_1_pdf}) can be used to calculate the expected value of the order statistics using the method given in \cite{Barakat_SMA2004}. This implies that $E[t_j]'s$ and hence $\bar{F_j}s$ are available offline. By using these values $E_{DD}$ ($\approx E_1(N^1)$) can be approximated as,
\begin{equation}
\label{Eddanalysis1}
E_{DD}\approx E[t_{l^*}]+\frac{\beta-\bar{F_{l^*}}}{\delta^{l^*}_{1,FC}},
\end{equation}
\noindent where the first term on R.H.S.\ is the mean time till the mean of increments becomes positive at the fusion node while the second term indicates the mean time for $F_k$ to cross $\beta$ from $t_{l^*}$ onward.
\subsubsection{$P_{MD}/P_{FA}$ Analysis}
\label{subsub:ch3:pfa_analysis}
We provide analysis under $H_1$. $P_{FA}$ analysis is same as that of $P_{MD}$ analysis with obvious changes. When the thresholds at local nodes are reasonably large, according to Lemma \ref{lemma:ch3:VS_argument_reg_finiteness_st}, with a large probability local nodes are making the right decisions and $t_k$ can be taken as the order statistics assuming that all local nodes make the right decisions. Then for missed detection the dominant event is $P_1(N^0<t_1)$. Also for reasonable performance we should select thresholds such that $P_1(N^1<t_1)$ is small. Then
\begin{align}
\label{eq:ch3:pfa_analysis_new_arg}
P_{MD} &=P_{1}(N^0<N^1) \geq P_{1}(N^0<t_1, N^1> t_1)\nonumber\\
& \approx P_{1}(N^0 <t_1).
\end{align} 
Under the above conditions, this lower bound should give a good approximation. In the following, we get an approximation for this.

Let $\xi_k=\log\left[g_{\mu_1}\left(Y_{k}\right)/g_{-\mu_0}\left(Y_{k} \right)\right]$. Then $F_k=\xi_1+\xi_2+...+\xi_k$ and if we assume that $\xi_k$ before $t_1$, has mean zero and has distribution symmetric about zero (e.g., $\sim \mathcal{N}(0,\sigma^2)$) then, 
\begin{IEEEeqnarray*}{rCl}
\IEEEeqnarraymulticol{3}{l}{P_{1}(\text{reject $H_1$ before $t_1$})}\nonumber \\
&\approx &\sum_{k=1}^\infty P_1\Big[\{F_k<-\beta\} \cap_{n=1}^{k-1} \{F_n >-\beta \}\big|t_1>k\Big]P[t_1>k]\\
&=&\sum_{k=1}^\infty\Big(P_1\Big[F_k<-\beta|\cap_{n=1}^{k-1} \{F_n >-\beta\}\Big]\\
&& \qquad \quad P_1\big[\cap_{n=1}^{k-1} \{F_n >-\beta\}\big]\Big) \Big(1-\Phi_{t_1}(k)\Big)\\
&\stackrel{(A)}=&\sum_{k=1}^\infty\Big(P_1[F_k<-\beta|F_{k-1} >-\beta]\\
&& \qquad \quad P_1[\inf_{1\leq n\leq k-1} F_n>-\beta]\Big)\Big(1-\Phi_{t_1}(k)\Big)\\
&\stackrel{(B)}\geq&\ \sum_{k=1}^\infty\Big(\int_{c=0}^{\infty}P_1[\xi_k<-c] f_{F_{k-1}}\{-\beta+c\}\,dc\Big) \\
&& \qquad \quad \Big(1-2P_1[F_{k-1}\leq -\beta]\Big)\Big(1-\Phi_{t_1}(k)\Big),
\end{IEEEeqnarray*}
\noindent where $\Phi_{t_1}$ is the Cumulative Distribution Function of $t_1$. Since we are considering only \{$F_k,k\leq t_1$\}, we remove the dependencies on $t_1$. In the above equations (A) is because of the Markov property of the random walk and (B) is due to the following lemma. This lemma can be obtained from \cite[p.\ 525]{Billingsley_PM_Book}.

\begin{lemma}
If $\xi_1$ has mean zero and distribution symmetric about zero, 
\[\pushQED{\qed} P\Big[\inf_{1\leq n\leq k-1} F_n>-\theta\Big] \geq 1-2P\Big[F_{k-1} \leq -\theta\Big]. \qedhere \popQED\]
\end{lemma}

Similarly we can write an upper bound by replacing $P[\cap_{n=1}^{k-1}\{F_n>-\theta\}]$ with $P[F_{k-1}>-\theta$]. We can make the lower bound tighter if we do the same analysis for the random walk between ${t_1}$ and $t_2$ with appropriate changes and add to the above bounds.
\subsubsection{Example 1}
\label{subsub:ch3:eg1_analysis}
We apply the DualSPRT on the following example and compare the $E_{DD}$ and $P_{FA}$ via analysis provided above with the simulation results. We assume that ${f_0}$ and ${f_1}$ are Gaussian with different means. This model is relevant when the noise and interference are log-normally distributed (\cite{Unnikrishnan_JSTSP2008}), and when $X_{k,l}$ is the sum of energy of a large number of observations at the secondary nodes at a low SNR.

Parameters used for simulation are as follows: $L=5$, $f_0\sim \mathcal{N}(0,1)$ and $f_1\sim\mathcal{N}(1,1)$. Also $f_0=f_{0,l}$ and $f_1=f_{1,l}$ for $1\leq l \leq L$, and $b=1$. We plot $P_E$ (=$P_{FA}$ under $H_0$ and $P_{MD}$ under $H_1$) and $E_{DD}$ (=$E_1[N]$ or $E_0[N]$) versus $\beta$ in Figure \ref{fig_dualSPRT_perf-comp_same-SNR}. Here $\gamma$, $\mu$ and $b$ are fixed for ease of calculation and they are chosen to provide good performance for the given $P_{MD}/P_{FA}$. The figure also contains the results obtained via analysis. We see a good match in theory and simulations. For comparison, Figure \ref{fig_dualSPRT_perf-comp_same-SNR} also contains asymptotic results which are presented in Section \ref{sec:DualSPRT_asy_opt} below.
\begin{figure}[h]
\vspace{-0.2 cm}
\subfloat[Probability of error]{\hspace{-0.3 cm}\includegraphics[trim = 16mm 0 19mm 5mm, clip=true,scale=0.29]{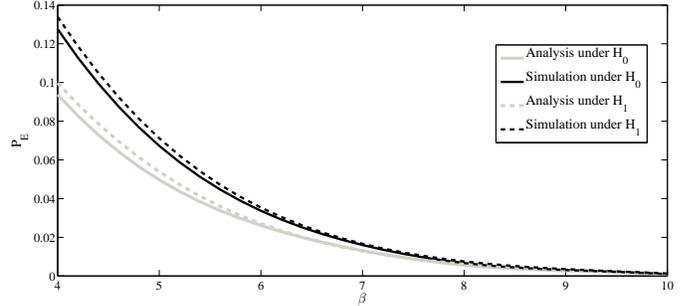}} \\
\subfloat[Expected detection delay]{\includegraphics[trim = 16mm 0 19mm 5mm, clip=true,scale=0.34]{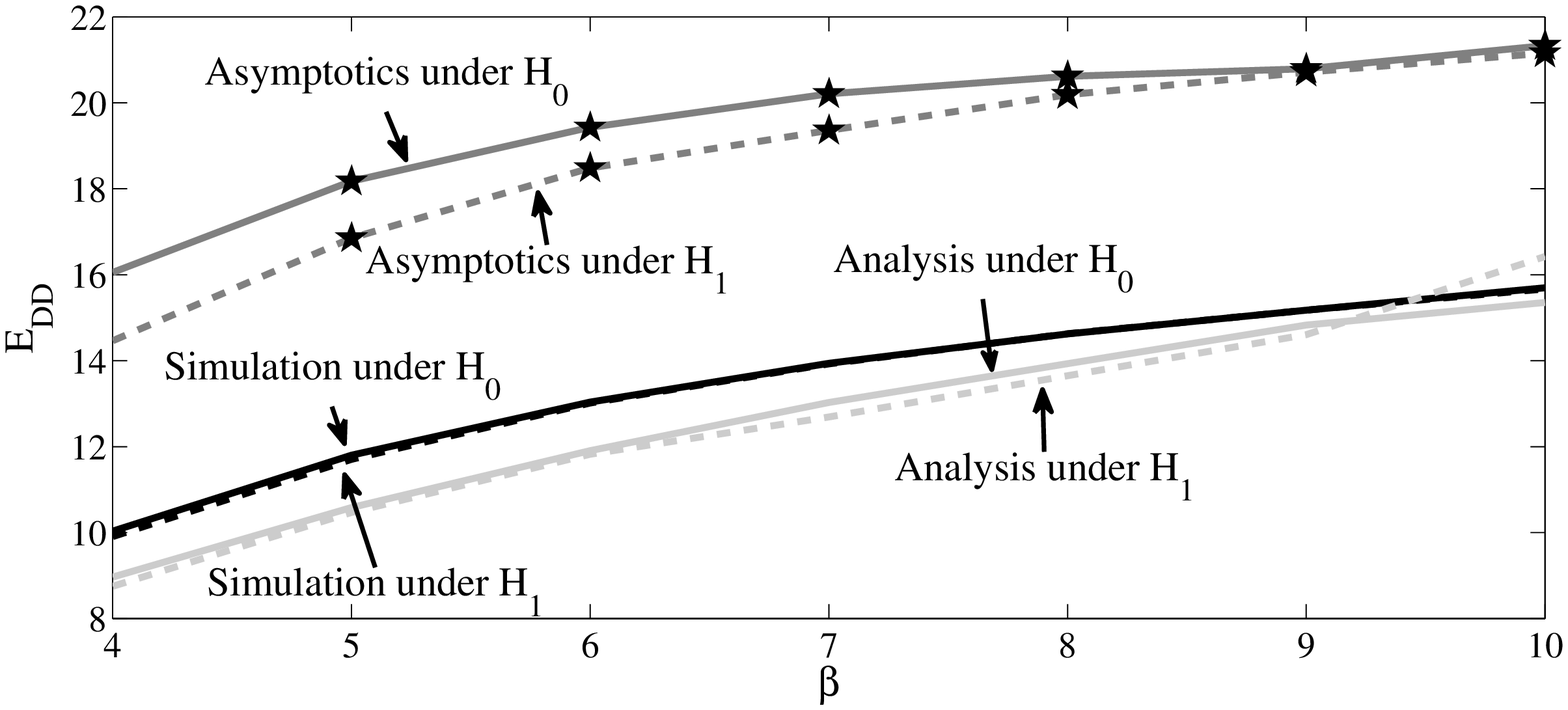}\label{fig_dualSPRT_perf-comp_same-SNR-2}}
\vspace{-0.2 cm} \caption{DualSPRT-Comparison between theory and simulation} 
\vspace{-0.1 cm}
\label{fig_dualSPRT_perf-comp_same-SNR}
\end{figure}

The above example is for the case when $X_{k,l}$ have the same distribution for different $l$ under the hypothesis $H_0$ and $H_1$. However in practice the $X_{k,l}$ for different local nodes $l$ will often be different because their receiver noise can have different variances and/or the path losses from the primary transmitter to the secondary nodes can be different. An example is provided here to illustrate the application of the above analysis to such a scenario. Now the order statistics $t_{l^*}$ in \eqref{Eddanalysis1} needs to be appropriately computed.
\subsubsection{Example 2}
\label{subsub:ch3:eg2_analysis}
There are five secondary nodes with primary to secondary channel gain being $0, -1.5, -2.5, -4$ and $-6$ dB respectively (corresponding post change means are $1, 0.84, 0.75, 0.63, 0.5$). $f_0\sim \mathcal{N}(0,1),f_0=f_{0,l}$ for $1\leq l \leq L$. Figure \ref{fig_dualSPRT_perf-comp_diff-SNR} provides the $E_{DD}$ and $P_{FA}$ via analysis and simulations. We see a good match.
\begin{figure}[h]
\vspace{-0.2 cm}
\subfloat[Probability of false alarm]{\hspace{-0.4 cm}\includegraphics[trim = 16mm 0 19mm 5mm, clip=true,scale=0.29]{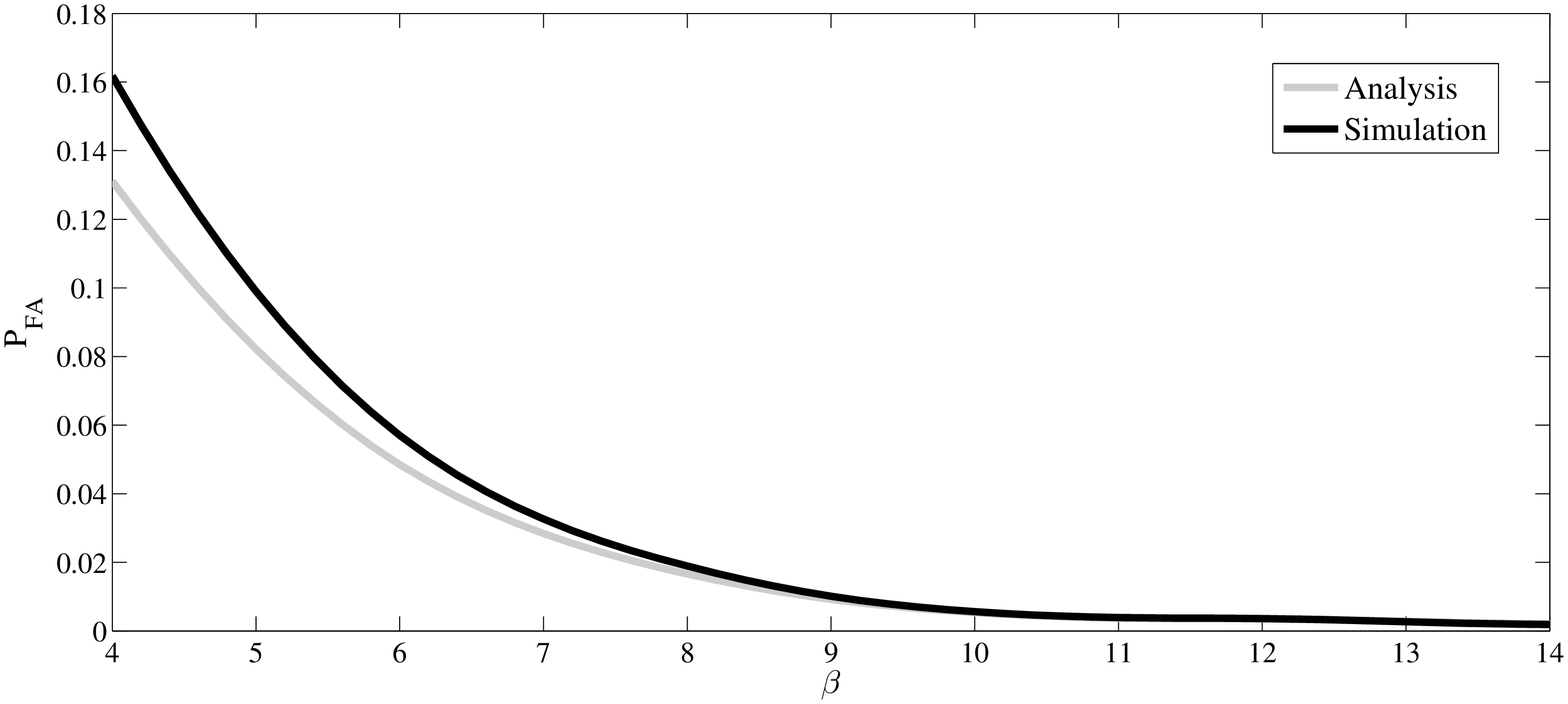}} \\
\subfloat[Expected detection delay]{\hspace{-0.4 cm}\includegraphics[trim = 16mm 0 19mm 5mm, clip=true,scale=0.29]{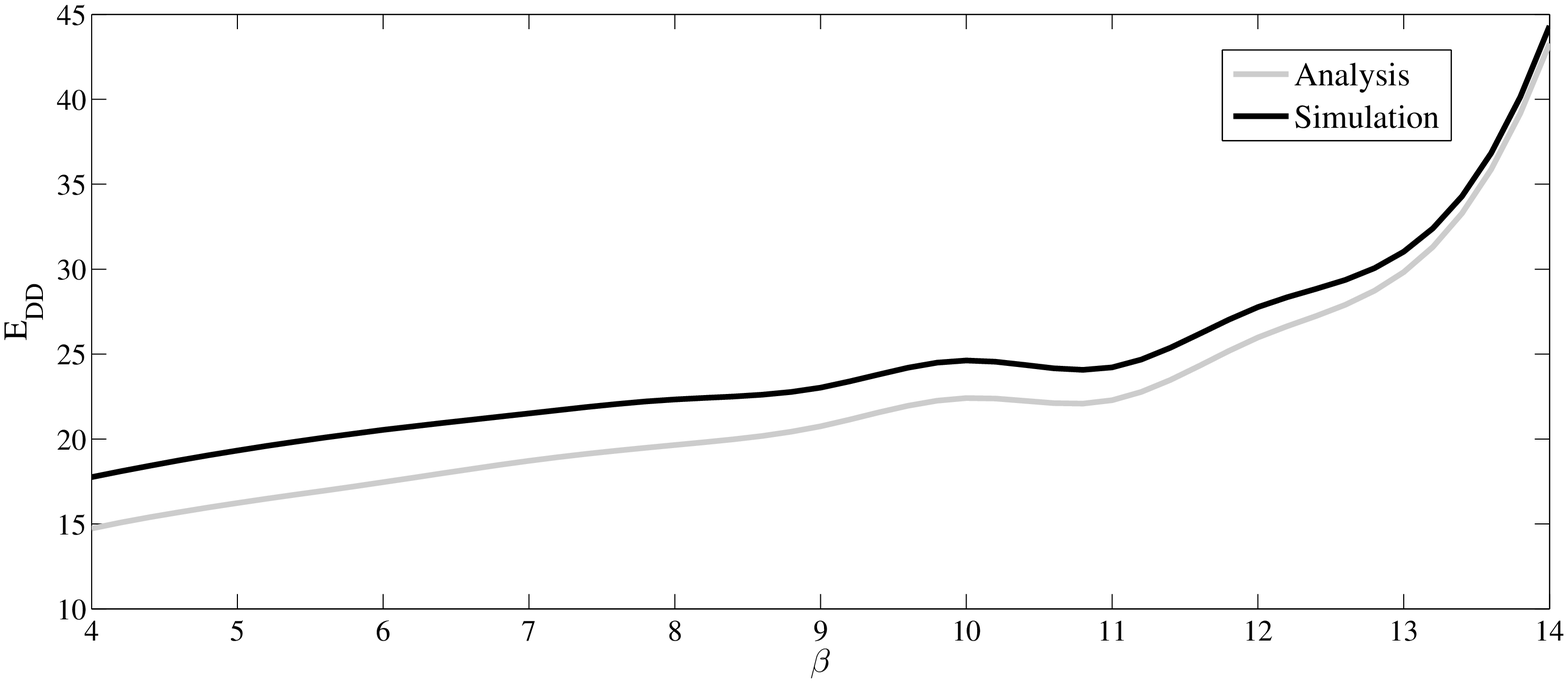}}
\vspace{-0.2 cm} \caption{DualSPRT-Comparison between theory and simulation} 
\vspace{-0.1 cm}
\label{fig_dualSPRT_perf-comp_diff-SNR}
\end{figure}

Unlike normal SPRT for i.i.d.\ observations where bounds are available on error probabilities based on the thresholds, here, in DualSPRT, the desired error probabilities are complicated functions of $\gamma_{0,l}$, $\gamma_{1,l}$, $\beta_0$, $\beta_1$, $b_0$, $b_1$, $\mu_1$ and $\mu_0$. From the analysis provided in this section, it is possible to provide beforehand, atleast approximately, the set of values for thresholds to achieve desired error probabilities and these can be used to design the test.
\section{Asymptotic Properties of DualSPRT}
\label{sec:DualSPRT_asy_opt}
In this section we prove asymptotic properties of DualSPRT.

We use the following notation:
\[D_{tot}^0=\sum_{l=1}^{L} D(f_{0,l}||f_{1,l}),\, D_{tot}^1=\sum_{l=1}^{L} D(f_{1,l}||f_{0,l}),\] \[r_l=\frac{D(f_{0,l}||f_{1,l})}{D_{tot}^0},\, \rho_l=\frac{D(f_{1,l}||f_{0,l})}{D_{tot}^1}.\]

Let $\mathcal{A}^i$ be the event that all the secondary users transmit $b_i$ when the true hypothesis is $H_i$. Also let $\Delta(\mathcal{A}^i)$ be the mean of increments of $F_k$ when $\mathcal{A}^i$ happens, i.e., $\Delta(\mathcal{A}^i)=E_i \left[(\log \frac{g_{\mu_1}(Y_k)}{g_{-\mu_0}(Y_k)})|\mathcal{A}^i\right]$. We use $\theta_i$ as the mean of the increments of $F_k$ when all the local nodes transmit wrong decisions under $H_i$.

In the rest of this section, local node thresholds are $\gamma_{0,l}=-r_l|\log c|, \gamma_{1,l}=\rho_l|\log c|$ and fusion center thresholds are $\beta_0=-{|\log c|}, \beta_1={|\log c|}$.

We will also need,
{\allowdisplaybreaks
\begin{eqnarray}
\label{eq:ch3:proof_notn_tau}
\tau_l(c)&\stackrel{\Delta}=&\sup \left\{n\geq1:W_{n,l} \geq  -r_l|\log c|\right\}, \nonumber \\
\tau(c)&\stackrel{\Delta}=&\max_{1 \leq l \leq L} \tau_l (c).
\end{eqnarray}}

Let $F_k^*$ be another likelihood ratio sequence (at FC) with expected value of its components as $\theta_i$ under $H_i$, the worst case value of the mean of the increments of $F_k$. Let the increments of $F_k^*$ be $\xi_1^*,\ldots,\xi_k^*$ which are i.i.d. Under $H_0$, $\theta_0 >0$ and under $H_1$, $\theta_1<0$. 

\begin{theorem}
\label{thm:ch3:edd}
For all  $l$ and for some $\alpha >1$, let $E_i[{|\xi_1^*|}^{\alpha+1}]<\infty$ and $E_i \left[{\Big|\log \frac{f_{1,l}(X_{1,l})}{f_{0,l}(X_{1,l})}\Big|}^{\alpha+1}\right]< \allowbreak \infty$, $i=0,1$. Then, under $H_i$,
\begin{equation*}
\overline{\lim_{c \to 0}} \frac{N}{|\log c|} \leq \frac{1}{D_{tot}^i}+M_i \text{ a.s.\ and in $L_1$,} 
\end{equation*}
where $M_i=C_i/\Delta(\mathcal{A}^i)$, $C_0=-\left[1+\frac{E_0[|\xi_1^*|]}{D_{tot}^0} \right]$ and $C_1=\left[1+\frac{E_1[|\xi_1^*|]}{D_{tot}^1} \right]$.
\end{theorem}

\begin{IEEEproof}
See Appendix \ref{proof:thm:ch3:edd}.
\end{IEEEproof}
\vspace{0.3 cm}

Figure \ref{fig_dualSPRT_perf-comp_same-SNR-2} compares the asymptotic upper bounds of $E_{DD}$ in Theorem \ref{thm:ch3:edd} with the approximations provided in Section \ref{sec:DualSPRT_per_ana} and simulations. We see that the approximate analysis of Section \ref{sec:DualSPRT_per_ana} provides much better approximation at threshold values of practical interest in Cognitive Radio. Perhaps this is the reason, the asymptotically optimal schemes do not necessarily provide very good performance at operating points of practical interest.

Next we consider the asymptotics of $P_{FA}$ and $P_{MD}$. Let $\displaystyle R_i=\min_{1\leq l \leq L} \Big(-\log \inf_{t \geq 0} E_i\Big[\exp \left({-t \log \frac{f_{1,l(X_{1,l})}}{f_{0,l}(X_{1,l})}}\right)\Big] \Big)$. 

Let $G_i$ and $\widehat{G}_i$ be the distributions of $|\xi_1^*|$ and $\xi_1^*$ respectively. Also let $g_i$ and $\widehat{g}_i$ be the corresponding moment generating functions. Let $\Lambda_i(\alpha)=\sup_{\lambda}(\alpha \lambda-\log g_i(\lambda))$, $\widehat{\Lambda}_i(\alpha)=\sup_{\lambda}(\alpha \lambda-\log \widehat{g}_i(\lambda))$ and take $\alpha^{+}_i=\text{ess }\sup |\xi_1^*|$. Let
\begin{equation}
\label{eq:thm_pe_s-param}
s_i(\eta)=\left\{
\begin{array}{ll}
\frac{\eta}{\alpha_i^{+}} &, \text{ if } \eta \geq \Lambda_i(\alpha^{+}_i),\\
\frac{\eta}{\Lambda^{-1}_i(\eta)} &, \text{ if } \eta \in (0, \Lambda_i(\alpha^{+}_i)).
\end{array}
\right.
\end{equation}
\begin{theorem}
\label{thm:ch3:pe-bayes}
Let $g_i(\lambda)<\infty$ in a neighbourhood of zero. Then,
\begin{enumerate}
\item[(a)]\label{thm:ch3:peb} $\displaystyle \lim_{c \to 0}\, \frac{P_{FA}}{c}=0$ if for some $0<\eta<R_0$, $s_0(\eta) >1$.
\item[(b)]\label{thm:ch3:pec} $\displaystyle \lim_{c \to 0}\, \frac{P_{MD}}{c}=0$ if for some $0<\eta<R_1$, $s_1(\eta) >1$. 
\end{enumerate}
\end{theorem}
\begin{IEEEproof}
See Appendix \ref{proof:thm_pe}.
\end{IEEEproof}
\begin{Remarks}
\label{rem:asy_opt_DualSPRT_rem-contraction-pri}
When $\alpha_i^{+}=\infty$ which is generally true, $\Lambda_i(\alpha^{+}_i)=\infty$ (\cite{Borovkov_1995}) and in Theorem \ref{thm:ch3:pe-bayes}(a) and \ref{thm:ch3:pe-bayes}(b) we need to consider only $R_i<\Lambda_i(\alpha^{+}_i)$. Also $\Lambda_i$ can be computed from $\widehat{\Lambda}_i$ using Contraction principle in Large Deviation theory (\cite{Dembo_LDP}).
\end{Remarks}

\begin{Remarks}
In \cite[Lemma 1-Appendix A]{Banerjee_WCOM}, it is proved that log likelhood ratio converts a large class of distributions into light tailed distributions and then $g_i(\lambda)$ is finite in a neighbourhood of zero. For instance, consider a regularly varying distribution for $Z_k$, $P(Z_k>x)=l'(x)x^{-\alpha}$, where $l'(x)$ is a slowly varying function and $\alpha>0$. Then,
$\log g_{\mu_1}(x)/g_{-\mu_0}(x)=\Big(l(x-\mu_1)/l(x+\mu_0)\Big) (x-\mu_1)^{-\alpha}(x+\mu_0)^{\alpha}$ $\leq x^{\beta_1+\alpha*\beta_2}$ for large $x$, any $\beta_1>0$ and an appropriately chosen $\beta_2>0$. This proves the conditions for \cite[Lemma 1]{Banerjee_WCOM} and hence exponential tail for $\widehat{G}_{i}(t)$ follows.
\end{Remarks}
We compare the asymptotic results obtained in Theorems \ref{thm:ch3:edd} and \ref{thm:ch3:pe-bayes} with that of SPRT with all the data available at the local nodes centrally without noise. Let $N_{ct}$ be the stopping time of such an SPRT. Then, from \cite[Theorem 2.11.1 and 2.11.2]{Govindarajulu_SS},
\begin{equation}
\label{eq:asy_opt_arg_edd}
\lim_{c \to 0} \frac{E_i[N_{ct}]}{|\log c|}=\frac{1}{D_{tot}^i}
\end{equation}
\begin{equation}
\label{eq:asy_opt_arg_pe}
\lim_{c \to 0} \frac{\log 1/P_{FA}}{|\log c|} \to 1, \lim_{c \to 0} \frac{\log 1/P_{MD}}{|\log c|} \to 1.
\end{equation}
Theorem \ref{thm:ch3:pe-bayes} implies the asymptotics \eqref{eq:asy_opt_arg_pe} on $P_{FA}$ and $P_{MD}$ for DualSPRT. Comparing Theorem \ref{thm:ch3:edd} with \eqref{eq:asy_opt_arg_edd}, we see that the rates of convergence of DualSPRT are optimal. For the limits to equal, we need $M_0$ and $M_1$ to be zero. In Section \ref{subsec:asy_prop_ex_gau} we compute $M_0$ and $M_1$ for Gaussian fusion center noise.

We can consider the asymptotic performance in the Bayesian framework also. Then the two hypotheses $H_0$ and $H_1$ are assumed to have known prior probabilities $\pi$ and $1-\pi$ respectively. A cost $c$ $(\geq 0)$ is assigned to each time step taken for decision. Let $W_i >0, i=0,1$ be the cost of falsely rejecting $H_i$. Then Bayes risk of a test $\delta$ with stopping time $N$ is defined as,
\begin{multline}
\label{eq:ch3:bayes_risk}
\mathcal{R}_c (\delta) =\pi [c {E}_0 (N)+W_0 {P}_0 \{reject~H_0\}]\\
+(1-\pi)[c {E}_1(N)+W_1 {P}_1\{reject H_1\}].
\end{multline}
Optimising \eqref{eq:ch3:bayes_risk} makes sense even when one does not have prior $\pi$ (i.e., within the frequentist framework) because then taking $\pi$ and $W_i$ appropriately, one can think of selecting a decision rule that asymptotically minimizes a weighted sum of $E_i[N]$ and $P_i[\text{reject }H_i], i=1,2$. 

Let $\mathcal{R}_c(\delta_{cent.})$ and $\mathcal{R}_c(\delta_{DualSPRT})$ be the Bayes's Risk of the optimal centralized SPRT without considering fusion center noise and of DualSPRT respectively. Then, (\cite[p.~2076]{Mei_TIT2008}),
\[\lim_{c\to 0}\frac{\mathcal{R}_c(\delta_{cent.})}{c |\log c|} =\left(\frac{\pi}{D_{tot}^0}+\frac{1-\pi}{D_{tot}^1} \right).\]
\noindent From Theorem \ref{thm:ch3:edd} and Theorem \ref{thm:ch3:pe-bayes}(a) and \ref{thm:ch3:pe-bayes}(b), using (\ref{eq:ch3:bayes_risk}), for DualSPRT with fusion center noise,
\[\lim_{c\to 0}\frac{\mathcal{R}_c(\delta_{DualSPRT})}{{c |\log c|}} = \left(\frac{\pi}{D_{tot}^0}+\frac{1-\pi}{D_{tot}^1}+ C\right),\]
where $C=M_0 \pi +M_1 (1-\pi)$. The constant $C$ can be made arbitrarily small by making $M_0$ and $M_1$ small.

\subsection{Example-Gaussian distribution}
\label{subsec:asy_prop_ex_gau}
In the following we apply Theorems \ref{thm:ch3:edd} and \ref{thm:ch3:pe-bayes} when the fusion center noise is Gaussian $\mathcal{N}(0,\sigma^2_{FC})$. We take $\mu_1=\mu_0=\mu>0$ and $b_1=-b_0=b>0$. For Theorem \ref{thm:ch3:edd}, $\Delta(\mathcal{A}^0)=-2\mu Lb/ \sigma^2_{FC}$ and $\Delta(\mathcal{A}^1)=2\mu Lb/ \sigma^2_{FC}$. Therefore $M_0$ and $M_1$ in Theorem \ref{thm:ch3:edd} $\to 0$ if $L \to \infty$ and/or $b \to \infty$. This also happens if $\sigma^2_{FC} \to 0$.

Using Remark \ref{rem:asy_opt_DualSPRT_rem-contraction-pri}, the condition in Theorem \ref{thm:ch3:pe-bayes}(a) is $\sigma^2_{FC}\eta/(4\mu^2\sqrt{2\eta}-2\mu L b)> 1$ for some $0<\eta<R_0$ and that for Theorem \ref{thm:ch3:pe-bayes}(b) is $\sigma^2_{FC}\eta/(4\mu^2\sqrt{2\eta}+2\mu L b)> 1$ for some $0<\eta<R_1$. Combining these two, it is sufficient to satisfy later condition with $0<\eta<\min(R_0,R_1)$. For Gaussian input observations at the local nodes, assuming $f_{1,l}=f_1$, $f_{0,l}=f_0$ for $1 \leq l \leq L$, we get $\delta_{i,l}=\delta_i$ and $\rho_{i,l}=\rho_i$, $\displaystyle  R_i=\frac{\delta_{i}^2}{2 \rho_{i}^2}$. This specifies upper-bounds for the choice of $\mu,L$ and $b$. 

\section{Improved Decentralized Sequential Tests: SPRT-CSPRT}
\label{sec:SPRT-CSPRT}
This section considers some improvements over DualSPRT. The improved algorithms are theoretically analysed and their performance is compared with existing decentralized schemes. 

\subsection*{New Algorithms: SPRT-CSPRT and DualCSPRT}
\label{sec:ch4:imp_DualSPRT}
In DualSPRT presented in Section \ref{sec:DualSPRT_dualsprt}, observations $\{Y_k\}$ to the fusion center are not always identically distributed. Till the first transmission from secondary nodes, these observations come from i.i.d.\ noise distribution, but not after that. Since the non-asymptotic optimality of SPRT is known for i.i.d. observations only (\cite{Siegmund_SATC_Book}), using SPRT at the fusion center is not optimal.

We improve DualSPRT with the following modifications. Steps (1)-(3) (corresponding to the algorithm run at the local nodes) are same as in DualSPRT. The steps (4) and (5) are replaced by:
\begin{enumerate}
\setcounter{enumi}{3}
\item Fusion center runs two algorithms:
\begin{IEEEeqnarray}{rCl}
\label{eq:ch4:FuseCUSUM_1}
 F^1_k &= &(F^1_{k-1} + \log \left[g_{\mu_1}\left(Y_{k}\right)\left / g_{Z}\left(Y_{k} \right. \right) \right])^+,\, F^1_0=0,\IEEEeqnarraynumspace\\
\label{eq:ch4:FuseCUSUM_0}
 F^0_k &= &(F^0_{k-1} + \log \left[g_{Z}\left(Y_{k}\right)\left / g_{-\mu_0}\left(Y_{k} \right. \right) \right])^-,\, F^0_0=0,\IEEEeqnarraynumspace
\end{IEEEeqnarray}
where $(x)^+=\max(0,x)$, $(x)^-=\min(0,x)$, $\mu_1$ and $\mu_0$ are positive constants, $g_Z$ is the pdf of i.i.d.\ noise $\{Z_k\}$ at the fusion center and $g_{\mu}$ is the pdf of $\mu+Z_k$.
\item The fusion center decides about the hypothesis at time \[\inf \{ k: F^1_k \geq \beta_1~or~F^0_k \leq -\beta_0\}\] and $\beta_0,\beta_1>0$. The decision is $H_1$ if $F^1_{k} \geq \beta_1$ and $H_0$ if $F_k^0 \leq -\beta_0$.
\end{enumerate}
 
The following discussion provides motivation for this test.
\begin{enumerate}
\item
If the SPRT sum defined in (\ref{eq:ch3:FuseCUSUM}) goes below zero it delays in crossing the positive threshold $\beta_1$. Hence if we keep SPRT sum at zero whenever it goes below zero, it reduces $E_{DD}$. 
This happens in CUSUM (\cite{Page_Biometrika1954}). Similarly one can use a CUSUM statistic under $H_0$ also. These ideas are captured in (\ref{eq:ch4:FuseCUSUM_1}) and (\ref{eq:ch4:FuseCUSUM_0}).

\item The proposed test is also capable of reducing false alarms caused by noise $Z_k$ before first transmission at $t_1$ from the local nodes. For $F_k^1$ and $F_k^0$ to move away from zero, the mean of increments should be positive and negative respectively. Let $\widehat{\mu}_k=E[Y_k]$ at time $k$. Then,  
\begin{equation}
E_{\widehat{\mu}_k} \left[\log \frac{g_{\mu_1}(Y_k)}{g_{Z}(Y_k)} \right]=D(g_{\widehat{\mu}_k}||g_{Z})-D(g_{\widehat{\mu}_k}||g_{\mu_1}).
\end{equation}

Hence before $t_1$, positive mean value of increments is not possible. After $t_1$ under $H_1$ (assuming the local nodes make correct decisions, the justification for which is provided in Section \ref{sec:DualSPRT}), the mean of increments becomes more positive. Similarly for $F_k^0$. But in case of DualSPRT, SPRT sum at the fusion center has the increments given by $\log \frac{g_{\mu_1}(Y_k)}{g_{-\mu_0}(Y_k)}$. This is difficult to keep zero only before $t_1$ and thus creates more errors due to noise $Z_k$.

\item Even though the problem under consideration is hypothesis testing, this is essentially a change detection problem at the fusion center. The observations at the fusion center have the distribution of noise before $t_1$ and after $t_1$ the mean changes. But in our scenario, this is a composite sequential change detection problem with the observations that are not i.i.d.\ and we look for change in both directions, it is difficult to use existing algorithms available for sequential change detection. Nevertheless our test (\eqref{eq:ch4:FuseCUSUM_1}-\eqref{eq:ch4:FuseCUSUM_0}) provides a guaranteed performance in this scenario.

\end{enumerate}

We consider one more improvement. When a local Cognitive Radio SPRT sum crosses its threshold, it transmits $b_1/b_0$. This node transmits till the fusion center SPRT sum crosses the threshold. If it is not a false alarm, then its SPRT sum keeps on increasing (decreasing). But if it is a false alarm, then the sum will eventually move towards the other threshold. Hence instead of transmitting $b_1$/ $b_0$ the Cognitive Radio can transmit a higher / lower value in an intelligent fashion. This should improve the performance. Thus we modify step (3) in DualSPRT as, 
\begin{IEEEeqnarray}{rCl}
\label{eq:ch4:quantisation1}
Y_{k,l}&=\sum_{i=1}^4 &b_i^1 \mathbb{I}\{W_{k,l} \in [\gamma_1+(i-1)\Delta_1, \gamma_1+i\Delta_1)\}+\nonumber \\ && b_i^0 \mathbb{I}\{W_{k,l} \in [-\gamma_0-(i-1)\Delta_1, -\gamma_0-i\Delta_0)\} \IEEEeqnarraynumspace
\end{IEEEeqnarray}
where $\Delta_1$ and $\Delta_0$ are the parameters to be tuned at the Cognitive Radio. $4 \Delta_1$ and $4\Delta_0$ are taken as $\infty$. The drift under $H_1$ ($H_0$) is a good choice for $\Delta_1$ ($\Delta_0$).

We call the algorithm with the above two modifications as SPRT-CSPRT (with `C' as an indication about the motivation from CUSUM). 

If we use CSPRT at both the secondary nodes and the fusion center with the proposed quantisation methodology (we call it DualCSPRT) it works better as we will show via simulations in Section \ref{sec:ch4:perf_comp}. In Section \ref{sec:ch4:perf_analysis} we will theoretically analyse SPRT-CSPRT. As the performance of DualCSPRT (Figure \ref{fig:ch4:comp_sprt-csprt_dualSPRT}) is close to that of SPRT-CSPRT, we analyse only SPRT-CSPRT.

\subsection{Performance Comparison}
\label{sec:ch4:perf_comp}
Throughout the rest of this section we use $\gamma_{1,l}=\gamma_{0,l}=\gamma$, $\beta_1=\beta_0=\beta$ and $\mu_1=\mu_0=\mu$ for the simplicity of simulations and analysis. 

We apply DualSPRT, SPRT-CSPRT and DualCSPRT on the following example and compare their $E_{DD}$ for various values of $P_{MD}$. We assume that the pre-change distribution $f_0$ and the post change distribution $f_1$ are Gaussian with different means.

For simulations we have used the following parameters. There are 5 nodes ($L=5$) and $f_{0,l}\sim \mathcal{N}(0,1)$, for $1\leq l \leq L$. Primary to secondary channel gains are $0$, $-1.5$, $-2.5$, $-4$ and $-6 dB$ respectively (the corresponding post change means of Gaussian distribution with variance 1 are 1, 0.84, 0.75, 0.63 and 0.5). We assume $Z_k\sim \mathcal{N}(0,5)$ and the mean of increments of DualSPRT and SPRT-CSPRT at the fusion center is taken as $2\mu Y_k$, with $\mu$ being 1. We also take $D_0=D_1=0$, $\{b^1_1,b^1_2,b^1_3,b^1_4\}=\{1,2,3,4\}$, $\{b^0_1,b^0_2,b^0_3,b^0_4\}=\{-1,-2,-3,-4\}$ and $b_1$=$-b_0$=1 (for DualSPRT). Parameters $\gamma$ and $\beta$ are chosen from a range of values to achieve a particular $P_{FA}$. Figure \ref{fig:ch4:comp_sprt-csprt_dualSPRT} provides the $E_{DD}$ and $P_{MD}$ via  simulations. We see a significant improvement in $E_{DD}$ compared to DualSPRT. The difference increases as $P_{MD}$ decreases. The performance under $H_0$ is similar. 

Performance comparisons with the asymptotically optimal decentralized sequential algorithms which do not consider fusion center noise (DSPRT \cite{Fellouris_TIT2011}, Mei's SPRT \cite{Mei_TIT2008}) are given in Figure \ref{fig:ch4:comp_DSPRT_MeiSPRT}. Note that DualSPRT and SPRT-CSPRT include fusion center noise. Here we take $f_{0,l}\sim \mathcal{N}(0,1)$, $f_{1,l}\sim \mathcal{N}(1,1)$ for $1\leq l \leq L$ and $Z_k\sim \mathcal{N}(0,1)$. We find that the performance of SPRT-CSPRT is close to that of DSPRT (which is second order asymptotically optimal) and better than Mei's SPRT (which is first order asymptotically optimal). Similar comparisons were obtained with other data sets.

\begin{figure}[h]
\vspace{0 cm}
\centering
\subfloat[Comparison among DualSPRT, SPRT-CSPRT and DualCSPRT for different SNR's between the Primary and the secondary users, under $H_1$.]{\includegraphics[trim = 16mm 0 19mm 5mm, clip=true,scale=0.34]{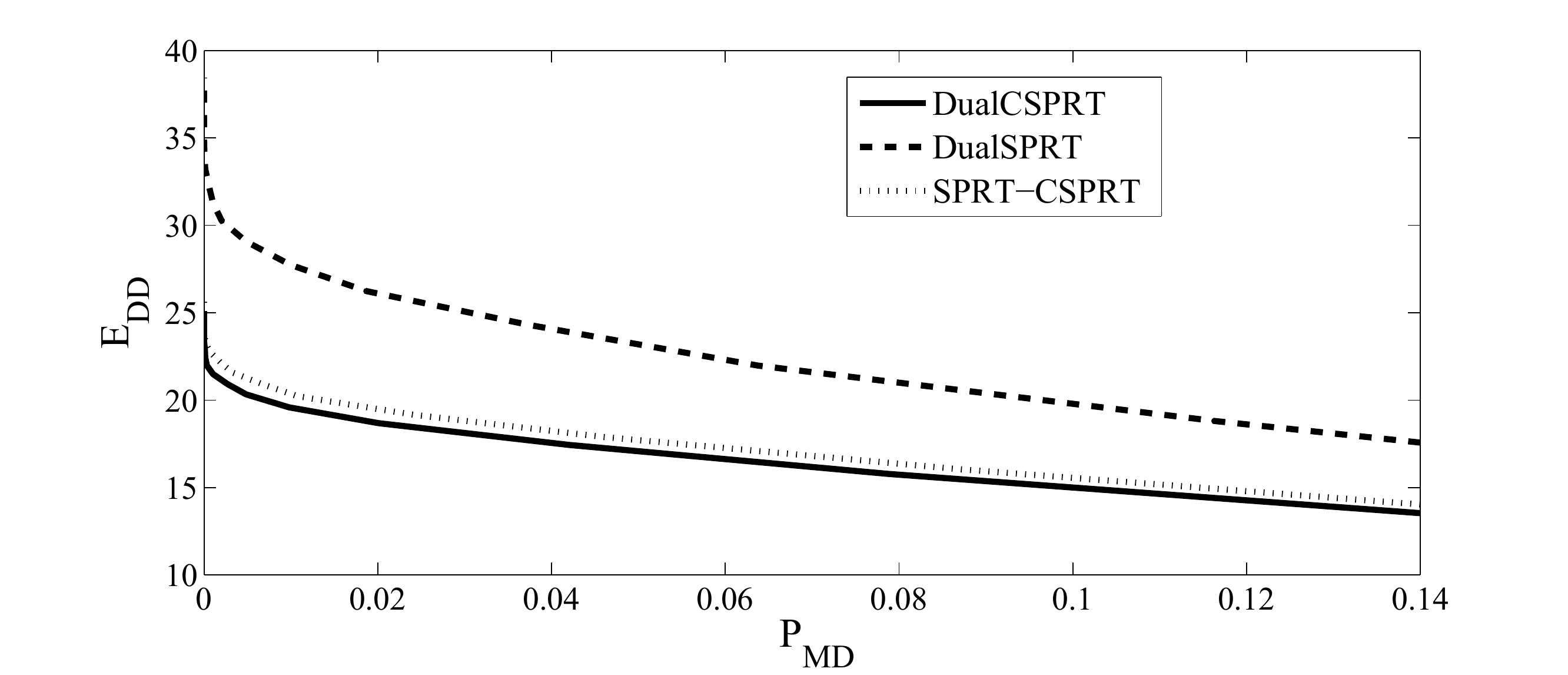} \label{fig:ch4:comp_sprt-csprt_dualSPRT}}\\
\subfloat[Comparison among DualSPRT, SPRT-CSPRT, Mei's SPRT and DSPRT under $H_1$.]{\includegraphics[trim = 16mm 0 19mm 5mm, clip=true,scale=0.32]{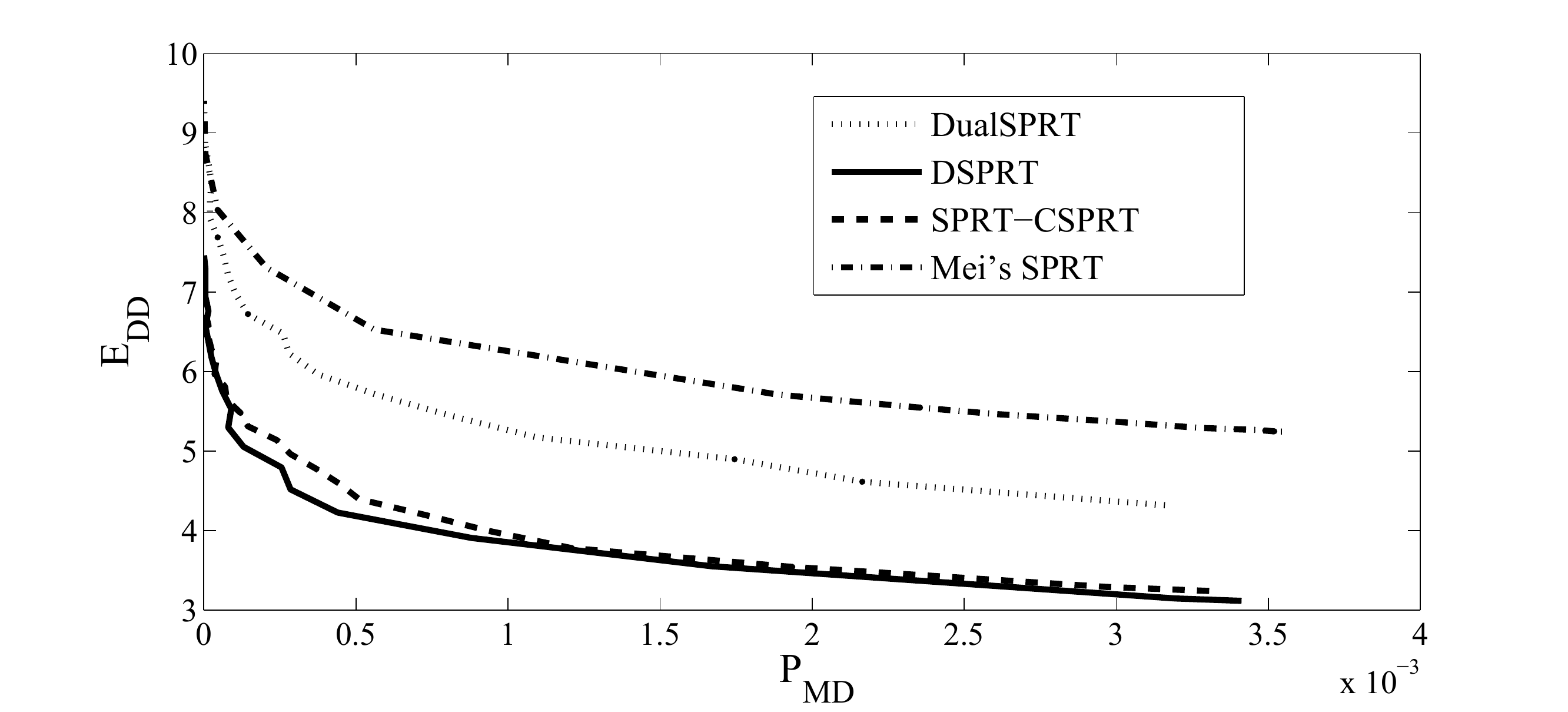}\label{fig:ch4:comp_DSPRT_MeiSPRT}}
\vspace{-0.2 cm} \caption{Comparisons with SPRT-CSPRT} 
\vspace{-0.1cm}
\end{figure}

\subsection{Performance Analysis of SPRT-CSPRT}
\label{sec:ch4:perf_analysis}
$E_{DD}$ and $P_{MD}/P_{FA}$ analysis is same under $H_1$ and $H_0$. Hence we provide analysis under $H_1$ only.
\subsubsection{$P_{MD}$ Analysis}
\label{subsec:ch4:pfa_analysis}
Between each change of mean of increments (which occurs due to the change in number of Cognitive Radios transmitting to the fusion node and due to the change in the value transmitted according to the quantisation rule (\ref{eq:ch4:quantisation1})) at the fusion center, under $H_1$, (\ref{eq:ch4:FuseCUSUM_1}) has a positive drift and behaves approximately like a normal random walk. Under $H_1$ (\ref{eq:ch4:FuseCUSUM_0}) also has a positive drift, but due to the $\min$ in its expression it will stay around zero and as the event of crossing negative threshold is rare (\ref{eq:ch4:FuseCUSUM_0}) becomes a reflected random walk between each drift change. Similarly under $H_0$, (\ref{eq:ch4:FuseCUSUM_1}) and (\ref{eq:ch4:FuseCUSUM_0}) become reflected random walk and normal random walk respectively. The false alarm occurs when the reflected random walk crosses its threshold. 

Under $H_1$, let
\begin{equation*}
\label{eq:ch4:FPT}
\tau_\beta  \stackrel{\triangle}{=} \inf \{ k \geq 1 : F^0_k \leq -\beta \} \text{ and }  T_\beta  \stackrel{\triangle}{=} \inf \{ k \geq 1 : F^1_k \geq \beta \}.
\end{equation*}
Following the argument in Section \ref{subsub:ch3:pfa_analysis} for $P_{MD}$, we get,
\begin{IEEEeqnarray}{rCl}
\label{eq:ch4:Pfa_analysis}
P_{1}(\text{reject $H_1$})& = &P_{1}(\tau_\beta < T_\beta)\approx P_{1}(\tau_\beta <t_1). \nonumber\\
P_{1}(\tau_\beta <t_1)&=&\sum_{k=1}^\infty P_{1}( \tau_{\beta} \leq k,k <t_1)\nonumber \\
&=&\sum_{k=1}^\infty P_{1}( \tau_{\beta} \leq k| k<t_1) P_{1}(t_1>k).\IEEEeqnarraynumspace
\end{IEEEeqnarray}
In the following we compute $P_{1}(\tau_\beta >  x|\tau_\beta <t_1)$ and $P_{1}(t_1>k)$. It is shown in \cite{Rootzen_AAP1988} that, 
\begin{eqnarray}
\label{eq:ch4:expoSensor} 
\lim_{\beta \to \infty} P_{1}\{\tau_\beta >  x|\tau_\beta <t_1\} = \exp(-\lambda_\beta x),\, x>0,
\end{eqnarray}
where $\lambda_\beta$ is obtained by finding solution to an integral equation obtained via renewal arguments (\cite{Ross_SP_Book}). Let $L(s)$ be the mean of $\tau_\beta$ with $F^0_0=s$ and $S_k=\log\left[g_{Z}\left(Y_{k}\right)/g_{-\mu_0}\left(Y_{k} \right)\right]$. Note that $\{S_k, k < t_1\}$ are i.i.d.\ From the renewal arguments, by conditioning on $S_1=z$, 
\begin{IEEEeqnarray*}{rCl}
\label{eq:ch4:FPTIntegralEqn}
L(s)& = &P(S_1>-s)(L(0)+1) \nonumber \\ && + \int_{-\beta-s}^{-s} (L(s+z)+1)\, dF_{S_1}(z)\, dz + F_{S_1}(-\beta-s),\nonumber
\end{IEEEeqnarray*}
\noindent where $F_S$ is the distribution of $S_k$ before the first transmission from the local nodes. This is a Fredholm integral equation of the second kind (\cite{Saaty_book}). Existence of a unique solution for it is shown in \cite{Banerjee_WCOM}. By solving these equations numerically, we get $\lambda_\beta=1/L(0)$.

From the central limit theorem approximation given in Section \ref{subsub:ch3:edd_analysis} we can find the distribution of $t_1$. Thus (\ref{eq:ch4:Pfa_analysis}) provides,
\[P_{1}(\text{False alarm before $t_1$})\approx\sum_{k=1}^\infty (1-e^{-\lambda_{\beta} k}) \prod_{l=1}^{L}(1-\Phi_{N_{l}}(k)),\]
\noindent where $\Phi_{N_{l}}$ is the Cumulative Distribution Function of $N_{l}$, obtained from the Gaussian approximation.
%
\subsubsection{$E_{DD}$ Analysis}
In this section we compute $E_{DD}$ theoretically. Recall that $t_i$ also approximates the first time at which $i$ local nodes are transmitting. Mean of $t_i$  can be computed from the method explained in \cite{Barakat_SMA2004}, for finding $k^{th}$ central moment of non i.i.d. $i^{th}$ order statistics. 

Between $t_{i}$ and $t_{i+1}$ the mean of the increments at the fusion center is not necessarily constant because there are four thresholds (each corresponds to different quantizations) at the secondary node. The transmitted value changes after crossing each threshold, $b_1^1\to b_2^1\ldots\to b_4^1$. Let $t^j_i, 1\leq j \leq 3$ be the time points at which a node changes the transmitting values from $b_j^1$ to $b_{j+1}^1$ between $t_{i}$ and $t_{i+1}$. We assume that with a high probability the secondary node with the lowest first passage time mean will transmit first, the node with the second lowest mean will transmit second and so on. This is justified by the fact that the distribution of the first passage time of $\gamma >0$ by a random walk with drift $\delta >0$ and variance $\sigma^2$ is $\mathcal{N}(\frac{\gamma}{\delta},\frac{\sigma^2\, \gamma}{\delta^3})$. Thus if $\delta$ is large, the mean $\gamma / \delta$ is small and the variance $\sigma^2 \gamma/ \delta^3$ is much smaller. In the following we will make computations under these approximations. The time difference between $t^{j th}_i$ and $t^{j th}_{i+1}$ transmission can be calculated if we take the second assumption (=$\Delta_1/\delta_{1,l}$). We know $E[t_i]$ for every $i$ from an argument given earlier. Suppose $l^{th}$ node transmits at $t_i^{th}$ instant and if $E[t_i]+\Delta_1/\delta_{1,l} < E[t_{i+1}]$ then $E[t^1_i]=E[t_i]+\Delta_1/\delta_{1,l}$. Similarly if $E[t^1_i]+\Delta_1/\delta_{1,l} < E[t_{i+1}]$ then $E[t^2_{i}]=E[t^1_i]+\Delta_1/\delta_{1,l}$ and so on. Let us represent the sequence $t=\{t_1,t^1_1,t^2_1,t^3_1,t_2,...,t^5_5\}$ (entry only for existing ones by the above criteria) by $T=\{T_1,T_2,T_3,...\}$.  

Let $\delta^k_{i,FC}$ be the mean of the increments at the fusion center between $T_k$ and $T_{k+1}$, under $H_i$. Thus $T_k$'s are the transition epochs at which the mean of the increments of fusion center changes from $\delta^{k-1}_{i,FC}$ to $\delta^{k}_{i,FC}$. Also let $\bar{F_k}=E[F_{T_k-1}]$ be the mean value of $F_k$ just before the transition epoch $T_k$. With the assumption of the very low $P_{fa}$ at the local nodes and from the knowledge of the sequence $t$ we can easily calculate $\delta^k_{1,FC}$ for each $T_k$. Similarly $\bar{F}_{k+1}=\bar{F_{k}}+\delta^{k}_{1,FC}(E[T_{k+1}]-E[T_{k}])$. Then,
\begin{equation}
\label{eq:ch4:Eddanalysis1}
E_{DD}\approx E[T_{l^*}]+\frac{\beta-\bar{F_{l^*}}}{\delta_{1,FC}^{l^*}}
\end{equation}
where 
\[l^*=min\{j:\delta_{1,FC}^j>0 \text{ and } \frac{\beta-\bar{F_j}}{\delta^j_{1,FC}}< E[T_{j+1}]-E[T_j]\}.\]

The above approximation of $E_{DD}$ is based on Central Limit Theorem and Law of Large Numbers and hence is valid for any distributions with finite second moments. 

Figure \ref{fig_ch4:table_analysis_sprt_csprt} provides the comparison between simulation and analysis. We used the same set-up as in Section \ref{sec:ch4:perf_comp} (with $Z_k\sim \mathcal{N}(0,1)$). We see a reasonable approximation.

\begin{figure}[h]
\vspace{-0.2 cm}
\subfloat[Probability of false alarm]{\hspace{-0.3 cm}\includegraphics[trim = 16mm 0 19mm 5mm, clip=true,scale=0.29]{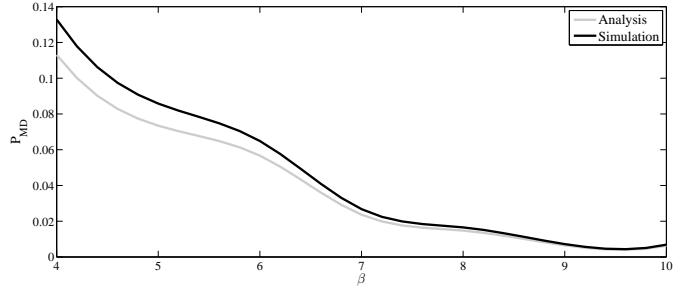}}\\
\subfloat[Expected detection delay]{\hspace{-0.3 cm}\includegraphics[trim = 16mm 0 19mm 5mm, clip=true,scale=0.29]{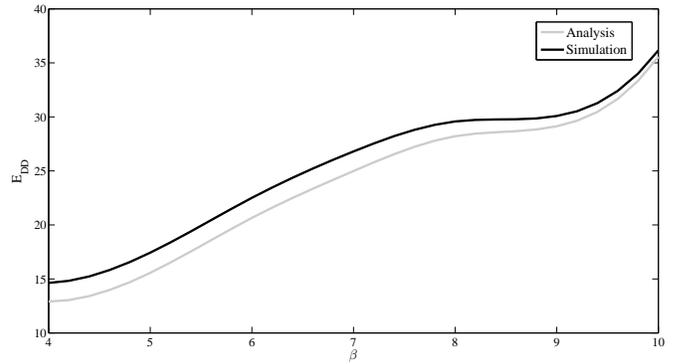}}
\vspace{-0.2 cm} \caption{Comparison of $E_{DD}$ and $P_{MD}$ obtained via analysis and simulation under $H_1$ for SPRT-CSPRT.} 
\vspace{-0.1 cm}
\label{fig_ch4:table_analysis_sprt_csprt}
\end{figure}

\section{Unknown Received SNRs and Fading}
\label{sec:SNR_uncertainty_fading}
This section considers the extensions of DualSPRT and SPRT-CSPRT to take care of the SNR uncertainty and the slow fading between the primary user and a Cognitive Radio. Since the transmissions from CR to FC are in CR network, we assume reporting channel to FC as AWGN only. This assumption is commonly made (\cite{Akyildiz_PC2011,Unnikrishnan_JSTSP2008}).
\subsection{Different and unknown SNRs}
\label{subsec:ch3:unk_rec_diff_SNR}
We consider the case where the received signal power from the PU to a CR node is fixed but not known to the local Cognitive Radio nodes. This can happen if the transmit power of the primary is not known and/or there is unknown shadowing. 
Now we limit ourselves to the energy detector where the observations $X_{k,l}$ are average energy of $M$ samples received by the $l^{th}$ Cognitive Radio node. Then for somewhat large $M$, the distributions of $X_{k,l}$ under $H_0$ and $H_1$ can be approximated by Gaussian distributions: $f_{0,l}\sim \mathcal{N}({\sigma}_l^2,2{\sigma}_l^4/M)$ and $f_{1,l}\sim \mathcal{N}(P_l+{\sigma_l}^2,2(P_l+{\sigma_l}^2)^2/M)$, where $P_l$ is the received power and $\sigma^2_l $ is the noise variance at the $l^{th}$ CR node. Under low SNR conditions ${(P_l+\sigma^2_l)}^2 \approx \sigma^4_l$ and hence $X_{k,l}$ are Gaussian distributed with mean change under $H_0$ and $H_1$.  Now taking $X_{k,l}-\sigma^2_l$ as the data for the detection algorithm at the $l^{th}$ node, since $P_l$ is unknown we can formulate this problem as a sequential hypothesis testing problem with
\begin{equation}
\label{eq:ch3:comphyp11}
H_0:\theta=0\,;\,H_1:\theta\geq \theta_1\,,
\end{equation}where $\theta$ is $P_l$ under $H_1$ and ${\theta_1}$ is appropriately chosen.

The problem
\begin{equation}
\label{eq:ch3:LaiSPRT_formln}
H_0 : \theta\leq\theta_0\,;\,H_1 : \theta\geq\theta_1\,,
\end{equation}
subject to
\[P_{\theta}\{\text{reject $H_0$}\}\leq\alpha\text{, for }\theta\leq\theta_0,\] \[P_{\theta}\{\text{reject $H_1$}\}\leq\beta\text{, for }\theta\geq\theta_1,\]
for exponential family of distributions is well studied in (\cite{Lai_AS1988}). The following algorithm of Lai \cite{Lai_AS1988} is asymptotically Bayes optimal and hence we use it at the local nodes instead of SPRT. Let $\theta\in A=[a_1,a_2]$. Define
\begin{eqnarray*}
W_{n,l}&=&max\left[\sum_{k=1}^n\log \frac{f_{\hat{\theta_n}}(X_k)}{f_{\theta_0}(X_k)},\sum_{k=1}^n\log \frac{f_{\hat{\theta_n}}(X_k)}{f_{\theta_1}(X_k)}\right],\\
N_l(g,c)&=&\inf\left\{n:W_{n,l}\geq g(n\,c)\right\},
\end{eqnarray*}
where $g()$ is a time varying threshold and $c>0$ is a design parameter. The function $g$ satisfies $g(t)\approx \log (1/t)$ as $t \to 0$ and is the boundary of an associated optimal stopping problem for the Wiener process (\cite{Lai_AS1988}). $\hat{\theta}_n$ is the Maximum-Likelihood estimate of $\theta$ bounded by $a_1$ and $a_2$. For Gaussian $f_0$ and $f_1$, ${\hat{\theta}}_n=\max\{a_1,\min[S_n/n,a_2]\}$. At time $N_l(g,c)$ decide upon $H_0$ or $H_1$ according as ${\hat{\theta}_{N_l(g,c)}\leq\theta^*}$ or ${\hat{\theta}_{N_l(g,c)}\geq\theta^*},$ where $\theta^*$ is obtained by solving $D(f_{\theta^*}||f_{\theta_0})=D(f_{\theta^*}||f_{\theta_1})$.

For our case where $H_0: \theta=0$, unlike in (\ref{eq:ch3:LaiSPRT_formln}) where $H_0: \theta \leq0$, $E_0[N_l(g,c)]$ largely depends upon the value $\theta_1$. As $\theta_1$ increases, $E_0[N_l(g,c)]$ decreases and $E_1[N_l(g,c)]$ increases. If $P_l\in [\underline{P}, \overline{P}]$ for all $l$ then a good choice of $\theta_1$, is $(\overline{P}-\underline{P})/2$.

\subsubsection{GLR-SPRT} 
First we modify DualSPRT. In the distributed setup with the received power at the local nodes unknown, the local nodes will use the Lai's algorithm mentioned above while the fusion node runs the SPRT. All other details remain same. We call this algorithm GLR-SPRT.
\subsubsection{GLR-CSPRT} 
This is a modified version of SPRT-CSPRT. Here, we modify GLR-SPRT to GLR-CSPRT with appropriate change in quantisation and using CSPRT at the fusion center instead of SPRT. The quantisation (\ref{eq:ch4:quantisation1}) is changed in the following way: if $\hat{\theta}_{N}\geq\theta^*$, 
let $\mathcal{I}_1=[g(k\,c),g(k\,c\,3\,\Delta))$, $\mathcal{I}_2=[g(k\,c\,3\,\Delta),g(k\,c\,2\,\Delta))$, $\mathcal{I}_3=[g(k\,c\,2\,\Delta),g(k\,c\,\Delta))$ and $\mathcal{I}_4=[g(k\,c\,\Delta),\infty)$. $Y_{k,l}=b^1_n  \text{ if }W_{k,l}\in \mathcal{I}_n$ for some $n$.
If $\hat{\theta}_{N}\leq\theta^*$ we will transmit from $\{b^0_1,b^0_2,b^0_3,b^0_4\}$ under the same conditions. Here, $\Delta$ is a tuning parameter and $0\leq 3\Delta \leq1$. 

The performance comparison of GLR-SPRT and GLR-CSPRT for the example in Section \ref{sec:ch4:perf_comp} (with $Z_k\sim \mathcal{N}(0,1)$) is given in Figure \ref{fig:ch4:glr-csprt}. Here $\Delta=0.25$. As the performance under $H_1$ and $H_0$ are different, we give the values under both. We can see that GLR-SPRT is always inferior to GLR-CSPRT. For $E_{DD}$ under $H_1$, interestingly GLR-CSPRT has lesser values than that of SPRT-CSPRT for $P_{FA}>0.02$ (note that SPRT-CSPRT has complete knowledge of the SNRs), while under $H_0$ it has higher values than SPRT-CSPRT.

\begin{figure}[h]
\vspace{-0.2 cm}
\centering \subfloat[Under $H_{1}$]{\includegraphics[trim = 16mm 0 19mm 5mm, clip=true,scale=0.33]{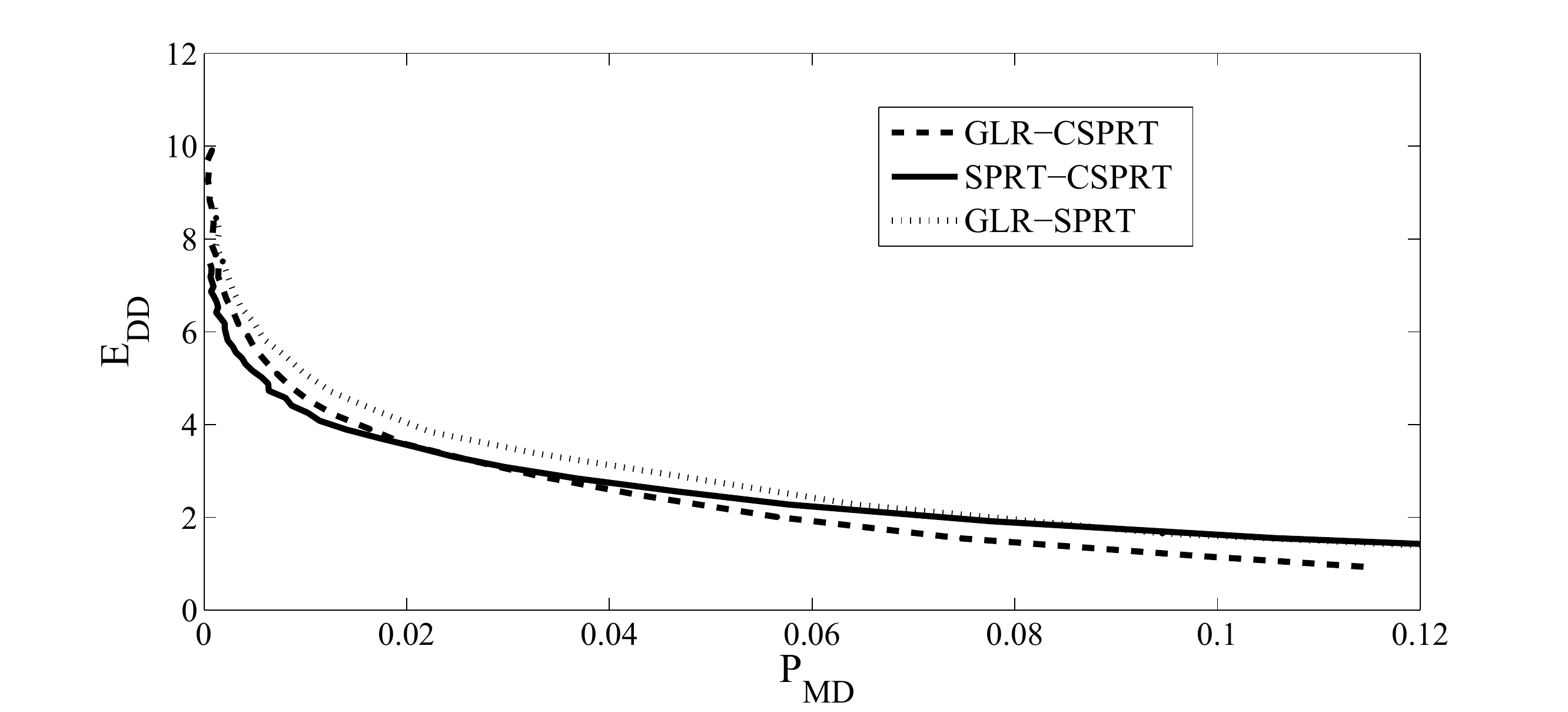}\label{fig:ch4:glr_1}}\\
\subfloat[Under $H_{0}$]{\includegraphics[trim = 16mm 0 19mm 5mm, clip=true,scale=0.35]{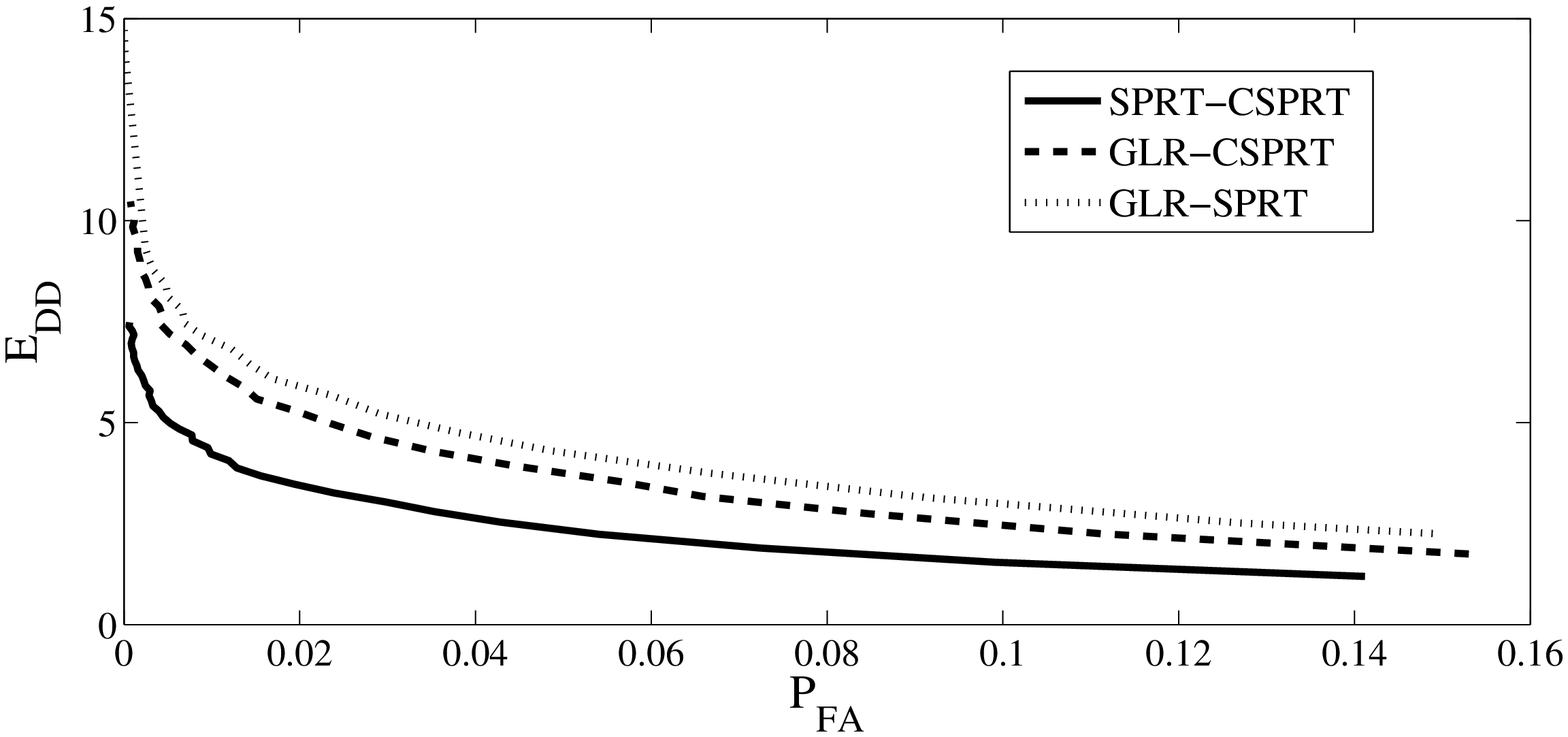}\label{fig:ch4:glr_0}}
\caption{Comparison among SPRT-CSPRT, GLR-SPRT and GLR-CSPRT for different SNR's between the Primary and the secondary users} 
\vspace{-0.2 cm}
\label{fig:ch4:glr-csprt}
\end{figure}
%
%
%
\subsection{Channel with Fading}
In this section we consider the system where the channels from the primary transmitter to the secondary nodes have fading $(h_l \neq 1)$. We assume slow fading, i.e., the channel coherence time is longer than the hypothesis testing time.

When the fading gain $h_l$ is known to the $l^{th}$ secondary node then this case can be considered as the different SNR case as in the example given in Section \ref{subsub:ch3:eg2_analysis}. Thus we consider the case where the channel gain $h_l$ is not known to the $l^{th}$ node. 

We consider the energy detector setup of Section \ref{subsec:ch3:unk_rec_diff_SNR}. However, now $P_l$, the received signal power at the local node $l$ is random. If the fading is Rayleigh distributed then $P_l$ has exponential distribution. The hypothesis testing problem becomes
\begin{equation}
\label{comphyp}
H_0:f_{0,l}\sim \mathcal{N}(0,{\sigma}^2); H_1:f_{1,l}\sim \mathcal{N}(\theta,{\sigma}^2)
\end{equation}
where $\theta$ is random with exponential distribution and ${\sigma}^2$ is the variance of noise. We will assume that $\sigma^2$ is known at the nodes.

We are not aware of this problem being handled via sequential hypothesis testing before. However we use Lai's algorithm in Section \ref{subsec:ch3:unk_rec_diff_SNR} where we take $\theta_1$ to be the median of the distribution of $\theta$, i.e., $P(\theta\geq\theta_1)=1/2$. This seems a good choice for $\theta_1$ as a compromise between $E_0[N]$ and $E_1[N]$.
\subsubsection{GLR-SPRT}
First we apply the technique on GLR-SPRT. We use an example where ${\sigma}^2=1, \theta=exp(1)$, Var($Z_k$)~=~1, and $L=5$. The performance of this algorithm is compared with that of DualSPRT (with perfect channel state information) in Figure \ref{fig_ch4:sf_case}. We observe that under $H_1$, for high $P_{MD}$ this algorithm works better than DualSPRT with channel state information, but as $P_{MD}$ decreases DualSPRT becomes better and the difference increases. For $H_0$, GLR-SPRT is always worse and the difference is almost constant.
\subsubsection{GLR-CSPRT}
Figure \ref{fig_ch4:sf_case} also provides comparison of DualSPRT, GLR-SPRT and GLR-CSPRT. Notice that the comment given for $E_{DD}$ for Figure \ref{fig:ch4:glr_1} is also valid here.
\begin{figure}[h]
\vspace{-0.2 cm}
\subfloat[Under $H_1$]{\hspace{-0.3 cm}\includegraphics[trim = 16mm 0 19mm 5mm, clip=true,scale=0.29]{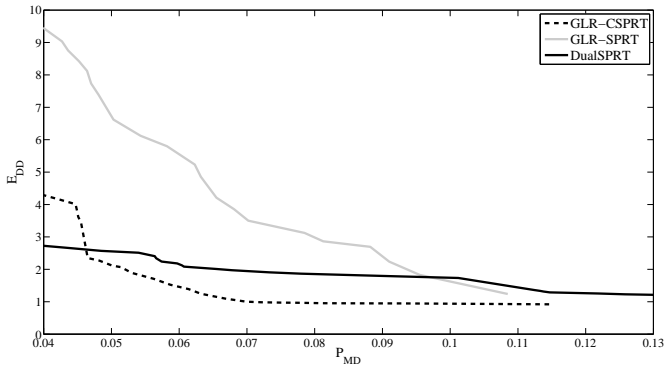} }\\
\subfloat[Under $H_0$]{\hspace{-0.3 cm}\includegraphics[trim = 16mm 0 19mm 5mm, clip=true,scale=0.29]{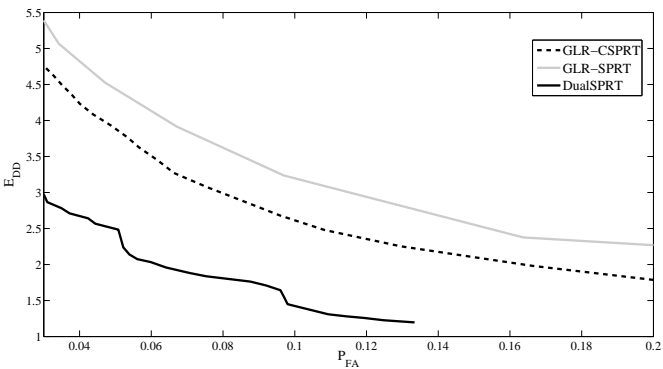}}
\vspace{-0.2 cm} \caption{Comparison among DualSPRT, GLR-SPRT and GLR-CSPRT with slow fading between the primary and the secondary users} 
\vspace{-0.2 cm}
\label{fig_ch4:sf_case}
\end{figure}

\section{Conclusions}
\label{sec:conclusions}
This paper presents fast algorithms for cooperative spectrum sensing satisfying reliability constraints. We have presented and analysed DualSPRT, a decentralized sequential hypothesis test. Simulation results corroborate the theoretical study of DualSPRT. Asymptotic properties of DualSPRT are also explored and its performance can approach asymptotically Bayes optimal tests. Improvement over DualSPRT using CUSUM statics for the fusion center test leads to another algorithm in which the selection of parameters is easy to choose apart from performance enhancement. We also provide approximate theoretical analysis of the algorithm. Numerical experiments show that this algorithm performs as well as an asymptotic order-2 optimal algorithm without fusion center noise, proposed in literature. We further extend our algorithms to cover the case of unknown SNR and channel fading and obtain satisfactory performance compared to perfect channel state information case. 
\appendices
\section{Proof of Theorem \ref{thm:ch3:edd}}
\label{proof:thm:ch3:edd}
We will prove the theorem under $H_0$. The proof under $H_1$ will follow in the same way.

Let $\nu(a)$ be the stopping time when a random walk starting at zero and formed by the sequence $\{\log \frac{g_{\mu_1(Z_k)}}{g_{-\mu_0(Z_k)}}+(\Delta(\mathcal{A}^0)-E_0[\log \frac{g_{\mu_1(Z_k)}}{g_{-\mu_0(Z_k)}}]), k \geq \tau(c)+1\}$ (with $-$ve drift under $H_0$) crosses $a$. Then,
\[N \leq N^0 \leq \tau(c)+\nu(-|\log c|-F_{\tau(c)+1}).\]
Therefore,
\begin{equation}
\label{eq:ch3:proof_e_0_0th_term}
\frac{N}{|\log c|} \leq \frac{\tau(c)}{|\log c|}+\frac{\nu(-|\log c|-F_{\tau(c)+1})}{|\log c|}.
\end{equation}
\indent We consider the first term on the R.H.S.\ of (\ref{eq:ch3:proof_e_0_0th_term}). From \cite[Remark 4.4, p.~90]{GUT_BOOK_2009} as $c \to 0$, $\tau_l(c) \to \infty$ a.s. and $\displaystyle \lim_{c \to 0}\frac{\tau_l(c)}{|\log c|}=-\frac{r_l}{\delta_{0,l}}=\frac{1}{D_{tot}^0}$ a.s. Therefore,
\begin{equation}
\label{eq:ch3:asy_pty_e_proof1}
\frac{\tau(c)}{|\log c|} \to \max_l \left\{-\frac{r_l}{\delta_{0,l}} \right\}=\frac{1}{D_{tot}^0} \text{ a.s.}
\end{equation}
Furthermore, from \cite[proof of Theorem 1 (i) $\Rightarrow$ (ii) p.~871]{Janson_AAP}, it can be seen that $\{\tau_l(c)/|\log c|\}$ is uniformly integrable for each $l$. Therefore, $\{\tau(c)/|\log c|\}$ is also uniformly integrable and hence,
\begin{equation}
\label{eq:proof_edd_last_ext_time_exp_converg}
\frac{E_0[\tau(c)]}{|\log c|} \to \frac{1}{D_{tot}^0}.
\end{equation}

The second term in R.H.S.\ of \eqref{eq:ch3:proof_e_0_0th_term},
\begin{equation}
\label{eq:ch3:proof_e_0_2nd_term}
\frac{\nu(-|\log c|-F_{\tau(c)+1})}{|\log c|} \leq \frac{\nu(-|\log c|)}{|\log c|}+\frac{\nu(-F_{\tau(c)+1})}{|\log c|}.
\end{equation}
\noindent We know, from \cite[Chapter III]{GUT_BOOK_2009}, as $c \to 0$
\begin{equation}
\label{eq:ch3:proof_e_0_3rd_term}
\frac{\nu(-|\log c|)}{|\log c|} \to- \frac{1}{\Delta(\mathcal{A}^0)} \text{ a.s.\ and in $L_1$}.
\end{equation}

Next consider $\nu(-F_{\tau(c)+1})$. Let $\widehat{F}_k^*$ be a random walk formed from $|\xi_k^*|$. It can be shown that $\widehat{F}_k^*$ stochastically dominates $F_k$ and thus we can make $\widehat{F}_k^* \geq F_k$ a.s.\ for all $k \geq 0$. Then,

\begin{IEEEeqnarray*}{rCl}
\frac{\nu(-F_{\tau(c)+1})}{|\log c|} &\leq & \frac{\nu(-\widehat{F}_{\tau(c)+1}^*)}{|\log c|}. \nonumber
\end{IEEEeqnarray*}
Also,
\begin{IEEEeqnarray*}{rCl}
\frac{\widehat{F}_{\tau(c)+1}^*}{|\log c|} &=& \frac{\widehat{F}_{\tau(c)+1}^*}{\tau(c)+1} \frac{\tau(c)+1}{|\log c|} \to E[|\xi_1^*|]\frac{1}{D_{tot}^0} \text{ a.s.} \IEEEeqnarraynumspace
\end{IEEEeqnarray*}
Thus,
\begin{eqnarray}
\label{eq:as_conv_entho}
\frac{\nu(-\widehat{F}^*_{\tau(c)+1})}{|\log c|}& = &\frac{\nu (-\widehat{F}^*_{\tau(c)+1})}{\widehat{F}^*_{\tau(c)+1}} \frac{\widehat{F}^*_{\tau(c)+1}}{|\log c|}\nonumber \\
&\to &\frac{-1}{\Delta(\mathcal{A}^0)}\frac{E[|\xi_1^*|]}{D_{tot}^0} \text{ a.s.}
\end{eqnarray}

From \eqref{eq:ch3:proof_e_0_0th_term}, \eqref{eq:ch3:asy_pty_e_proof1}, \eqref{eq:ch3:proof_e_0_2nd_term}, \eqref{eq:ch3:proof_e_0_3rd_term} and \eqref{eq:as_conv_entho},
\begin{equation*}
\overline{\lim_{c\to 0}} \frac{N}{|\log c|} \leq \frac{1}{D_{tot}^0}-\frac{1}{\Delta(\mathcal{A}^0)}\frac{E_0[|\xi_1^*|]}{D_{tot}^0} \text{ a.s.}
\end{equation*}

Now we show $L_1$ convergence.
For $\alpha>1$, 
{\allowdisplaybreaks
\begin{IEEEeqnarray}{rCl}
\lefteqn{\frac{E_0[{\nu(-\widehat{F}^*_{\tau(c)+1})}^{\alpha}]}{|\log c|^{\alpha}}}\nonumber\\
 &=& \frac{1}{|\log c|^{\alpha}} \int\limits_{0}^{|\log c|} E_0[{\nu(-x)}^{\alpha}|\widehat{F}^*_{\tau(c)+1}=x]\,dP_{\widehat{F}^*_{\tau(c)+1}}(x)\nonumber \\
& & +\,\frac{1}{|\log c|^{\alpha}} \int\limits_{|\log c|}^{\infty} E_0[{\nu(-x)}^{\alpha}]\,dP_{\widehat{F}^*_{\tau(c)+1}}(x)\nonumber \\
\label{eq:ch3:asy_pty_e_proof2}
&\leq & \frac{E_0[{\nu(-|\log c|)}^{\alpha}]}{|\log c|^{\alpha}}+\nonumber \\
& & \ \ \ \ \ \ \ \int\limits_{|\log c|}^{\infty} \frac{E_0[{\nu(-x)}^{\alpha}]}{x^{\alpha}}\, \frac{x^{\alpha}}{|\log c|^{\alpha}}\,dP_{\widehat{F}^*_{\tau(c)+1}}(x).\IEEEeqnarraynumspace 
\end{IEEEeqnarray}}
When $-$ve part of the increments of random walk of $\nu(t)$ has finite $\alpha^{\text{th}}$ moment (\cite[Chapter 3, Theorem 7.1]{GUT_BOOK_2009}), $E_0[{\nu(-x)}^{\alpha}]/{x^{\alpha}} \to {(-{1}/{\Delta(\mathcal{A}^0)})}^{\alpha}$ as $x \to \infty$. Thus for any $\epsilon >0$, $\exists\, M$ such that
\[\frac{E_0[{\nu(-x)}^{\alpha}]}{x^{\alpha}} \leq \left(\epsilon+{\left(\frac{-1}{\Delta(\mathcal{A}^0)}\right)}^{\alpha}\right) \text{ for $x > M$}.\]
Take $c_1$ such that $|\log c| >M$ for $c < c_1$. Then, for $c<c_1$,
\begin{IEEEeqnarray}{rCl}
\lefteqn{\int\limits_{|\log c|}^{\infty} \frac{E_0[{\nu(-x)}^{\alpha}]}{x^{\alpha}}\, \frac{x^{\alpha}}{|\log c|^{\alpha}}\,dP_{\widehat{F}^*_{\tau(c)+1}}(x)}\nonumber \\
& \leq & \frac{\epsilon+{\left(\frac{-1}{\Delta(\mathcal{A}^0)}\right)}^{\alpha}}{|\log c|^{\alpha}} \, \int\limits_{|\log c|}^{\infty} x^{\alpha}\,dP_{\widehat{F}^*_{\tau(c)+1}}(x) \nonumber\\
\label{eq:ch3:asy_pty_e_proof3}
&\leq & \frac{\epsilon+{\left(\frac{-1}{\Delta(\mathcal{A}^0)}\right)}^{\alpha}}{|\log c|^{\alpha}}\,E_0[(\widehat{F}^*_{\tau(c)+1})^{\alpha}].\IEEEeqnarraynumspace
\end{IEEEeqnarray}

Since $\lim_{c \to 0} \frac{\tau(c)}{|\log c|}=\frac{1}{D_{tot}^0}$ a.s. and $\{\frac{{ \tau(c)}^{\alpha}}{{|\log c|}^{\alpha} }\}$ is uniformly integrable, when $E_0 \left[ \left(\log \frac{f_{1,l}(X_{1,l})}{f_{0,l}(X_{1,l})} \right)^{\alpha+1} \right] < \infty,$ $1\leq l \leq L$ and $E[|\xi_1^*|^{\alpha+1}]<\infty$, we get, (\cite[Remark 7.2, p.~42]{GUT_BOOK_2009}),
\begin{equation*}
\lim_{c \to 0} \frac{E_0[(\widehat{F}^*_{\tau(c)+1})^{\alpha}]}{|\log c|^{\alpha}} = \frac{E[|\xi_1^*|^{\alpha}]}{D_{tot}^0},
\end{equation*}
and
\begin{equation}
\label{eq:ch3:proof_e_0_last_term}
\sup_{c >0} \frac{E_0[(\widehat{F}^*_{\tau(c)+1})^{\alpha}]}{|\log c|^{\alpha}} < \infty.
\end{equation}
From \eqref{eq:ch3:asy_pty_e_proof2}, \eqref{eq:ch3:proof_e_0_last_term}, for some $1>\delta>0$, 
\begin{IEEEeqnarray*}{rCl}
\lefteqn{\sup_{\delta > c >0}\frac{E_0[\nu(-\widehat{F}^*_{\tau(c)+1})^{\alpha}]}{|\log c|^{\alpha}}}\nonumber\\
&\leq & \sup_{\delta> c > 0} \frac{E_0[\nu(-|\log c|)^{\alpha}]}{|\log c|^{\alpha}} \\
& & \quad +\left[\epsilon+\left( \frac{-1}{\Delta(\mathcal{A}^0)} \right)^{\alpha} \right]\sup_{\delta> c > 0} \frac{E_0[(\widehat{F}^*_{\tau(c)+1})^{\alpha}]}{|\log c|^{\alpha}}\IEEEeqnarraynumspace\\
& < & \infty.
\end{IEEEeqnarray*}
Therefore, $\{\nu(-\widehat{F}^*_{\tau(c)+1})/|\log c| \}$ is uniformly integrable and hence, from \eqref{eq:as_conv_entho}, 
\begin{equation*}
\lim_{c \to 0} \frac{E_0[\nu(-\widehat{F}^*_{\tau(c)+1})]}{|\log c|} \leq -\frac{1}{\Delta(\mathcal{A}^0)}.\frac{E[|\xi_1|^*]}{D_{tot}^0}.
\end{equation*}
This, with (\ref{eq:ch3:proof_e_0_0th_term}), \eqref{eq:proof_edd_last_ext_time_exp_converg}, (\ref{eq:ch3:proof_e_0_2nd_term}) and \eqref{eq:ch3:proof_e_0_3rd_term}, implies that (since $\epsilon$ can be taken arbitrarily small),
{\allowdisplaybreaks
\begin{eqnarray*}
\overline{\lim_{c\to 0}} \frac{E_0[N]}{|\log c|} \leq \frac{1}{D_{tot}^0}+M_0,
\end{eqnarray*}}
where $ M_0=-\frac{1}{\Delta(\mathcal{A}^0)} \left[1+\frac{E_0[|\xi_1^*|]}{D_{tot}^0} \right]$. 

Similarly we can prove $ \overline{\lim}_{c\to 0} \frac{E_1[N]}{|\log c|}\leq\frac{1}{D_{tot}^1}+M_1$, where $M_1=\frac{1}{\Delta(\mathcal{A}^1)} \left[1+\frac{E_1[|\xi_1^*|]}{D_{tot}^1} \right].$ \hfill \QED

\section{Proof of Theorem \ref{thm:ch3:pe-bayes}}
\label{proof:thm_pe}
We prove the result for $P_{FA}$. For $P_{MD}$ it can be proved in the same way.

Probability of False Alarm can be written as,
\begin{IEEEeqnarray}{rCl}
\label{eq:ch3:proof_pfa_starting}
P_0(\text{Reject $H_0$})=P_0[\text{FA before $\tau(c)$}]+P_0[\text{FA after $\tau(c)$}].\IEEEeqnarraynumspace
\end{IEEEeqnarray}

Consider the first term in the R.H.S.\ of \ref{eq:ch3:proof_pfa_starting}. It can be shown that $F_k^*$ stochastically dominates $F_k$ under $H_0$. Thus we can construct $\{F_k^*\}$ such that $F_k^* \geq F_k$ a.s.\ for all $k \geq 0$ and hence
{\allowdisplaybreaks \begin{IEEEeqnarray}{rCl}
{P_0[\text{FA before $\tau(c)$}]}
&\leq &  P_0[ \sup_{0\leq k \leq \tau(c)} F_k^* \geq |\log c| ]\nonumber \\
&=& P_0[\sum_{k=0}^{\tau(c)}|\xi_k^*| \geq |\log c|] \label{eq:proof_pe_pfa_bef_tao(c)}.
\end{IEEEeqnarray}}

From \cite[Theorem 1.3]{Iksanov_ECP} $E_0[e^{\eta \, \tau_{l}(c)}]<\infty$, for $0< \eta<R^l_0$ and $ R^l_0=-\log \inf_{t \geq 0} E_0\Big[e^{-t \log \frac{f_{1,l(X_{1,l})}}{f_{0,l}(X_{1,l})}}\Big]$. Combining this fact with $\tau(c)< \sum_{l=1}^L \tau_l(c)$ and the fact that $\tau_l(c)$ are independent of each other (see (\ref{eq:ch3:proof_notn_tau})) yields $E_0[e^{\eta \tau(c)}] < E_0[e^{\sum_{l=1}^L \eta \tau_{l}(c)}]<\infty$, for $0< \eta < R_0=\min_l R^l_0$. Therefore, from Markov inequality, with $k_1=E_0[e^{\eta \tau(c)}]$,
\begin{equation}
\label{lemma:ch3:e_f_tau_tau-bound-1}
P[\tau(c)>t]\leq k_1 \exp (-\eta t).
\end{equation}
Let $\widehat{F}^*_n=\sum_{k=1}^n |\xi_k^*|$. Then, with \eqref{lemma:ch3:e_f_tau_tau-bound-1}, the expected value of $|\xi_k^*|$ being positive and with exponential tail assumption of $G_0(t)$, from \cite[Theorem 1, Remark 1]{Borovkov_1995}, \eqref{eq:proof_pe_pfa_bef_tao(c)} is,
\begin{equation}
\label{them_pe_borokov_thm}
P_0[\widehat{F}^*_{\tau(c)}>|\log c|]\leq k_2 \exp (-s_0(\eta) |\log c|),
\end{equation}
for any $0<\eta<R_0$. $k_2$ is a constant and $s_0(\eta)$ is defined in \eqref{eq:thm_pe_s-param}. 
Therefore, 
\begin{eqnarray}
\label{them_pe_borokov_thm_use}
\frac{P_0[\text{FA before $\tau(c)$}]}{c } \leq k_2 \frac{c^{s_0(\eta)}}{c} \to 0,
\end{eqnarray}
if $s_0(\eta)>1$ for some $\eta$.

Now we consider the second term in \eqref{eq:ch3:proof_pfa_starting},
\begin{IEEEeqnarray*}{rCl}
\lefteqn{P_0[\text{FA after $\tau(c)$}]}\nonumber \\
&=& P_0[\text{FA after $\tau(c)$};\mathcal{A}^0]+P_0[\text{FA after $\tau(c)$};(\mathcal{A}^0)^c]\IEEEeqnarraynumspace
\end{IEEEeqnarray*}
Since events $\{\text{FA after $\tau(c)$}\}$ and $(\mathcal{A}^0)^c$ are mutually exclusive, the second term in the above expression is zero. Now consider $P_0\left[\text{FA after $\tau(c)$};\mathcal{A}^0\right]$. For $0<r<1$,
{\allowdisplaybreaks
\begin{IEEEeqnarray}{rCl}
\lefteqn{P_0\left[\text{FA after $\tau(c)$};\mathcal{A}^0\right]}\nonumber \\
 & \leq &  P_0 \Big[\text{Random walk with drift $\Delta(\mathcal{A}^0)$} \nonumber\\
&& \qquad \quad \text{and initial value $F_{\tau(c)+1}$ crosses $|\log c|$}\Big] \nonumber\\
 & \leq & P_0 \Big[\text{Random walk with drift $\Delta(\mathcal{A}^0)$} \nonumber\\
 && \qquad \quad \text{and $F_{\tau(c)+1} \leq r |\log c|$ crosses $|\log c|$}\Big] \nonumber \\
 & & +\: P_0 \Big[\text{Random walk with drift $\Delta(\mathcal{A}^0)$}\nonumber\\
 && \qquad \quad \text{and $F_{\tau(c)+1} > r |\log c|$ crosses $|\log c|$}\Big] \nonumber\\
&\leq & P_0 \Big[\text{Random walk with drift $\Delta(\mathcal{A}^0)$}\nonumber\\
&& \qquad \quad \text{and $F_{\tau(c)+1} \leq r |\log c|$ crosses $|\log c|$}\Big] \nonumber \\
\label{eq:ch3:pfa_proof_1}
& & +\: P_0\left[F_{\tau(c)+1} > r |\log c|\right].
\end{IEEEeqnarray} }
\noindent Considering the first term in the above expression,
{\allowdisplaybreaks
\begin{IEEEeqnarray}{rCl}
\lefteqn{\Big (P_0 \big[\text{Random walk with drift $\Delta(\mathcal{A}^0)$}}\nonumber\\
\lefteqn{\qquad \quad \text{and $F_{\tau(c)+1} \leq r |\log c|$ crosses $|\log c|$}\big] \Big) /c}\nonumber \\
& \leq & \Big (P_0 \big[\text{Random walk with drift $\Delta(\mathcal{A}^0)$} \nonumber\\
&&\qquad \quad \text{and $F_{\tau(c)+1} = r |\log c|$ crosses $|\log c|$}\big] \Big) /c \nonumber \\
 \label{eq:ch3:proof_pfa_2nd_term_final}
 & \stackrel{(A)}\leq & \frac{\exp(-(1-r)|\log c|s')}{c}= \frac{c^{(1-r)s'}}{c} \to 0,
\end{IEEEeqnarray}}
iff  $(1-r) s' > 1$. Here $(A)$ follows from \cite[p.~78-79]{POO_BOOK_1}
\footnote{For a random walk $W_n=\sum_{i=1}^n X_i$, with stopping times $T_{a}=\inf\{n\geq 1:W_n \leq a \}$, $T_{b}=\inf\{n\geq 1:W_n \geq b\}$ and $T_{a,b}=\min(T_a,T_b)$, $a <0 <b$, let $s'$ be the non-zero solution to $M(s')=1$, where $M$ denotes the M.G.F.\ of $X_i$. Then, $s'<0$ if $E[X_i] >0$, and $s' >0$ if $E[X_i] <0 $ and $E[\exp(s'W_{T_{a,b}})]=1$ (\cite[p.~78-79]{POO_BOOK_1}). Then it can be shown that $P(W_{T_{a}}) \leq \exp(-s'a)$ when $E[X_i] >0$ and $P(W_{T_{b}}) \leq \exp(-s'b)$ when $E[X_i] <0$.}
 where $s'$ is positive and it is the solution of $E_0\Big[e^{s'\,\log \frac{g_{\mu_1}(Y_k)}{g_{-\mu_0}(Y_k)}}|\mathcal{A}^0\Big]=1.$ 
 
We choose $s'>1$ and $0<r<1$ to satisfy $(1-r)s' > 1$. 

Consider the second term in (\ref{eq:ch3:pfa_proof_1}). Using the stochastical dominance of $\{F_k\}$ by $\{\widehat{F}^*_k\}$,
\begin{eqnarray}
\label{eq:ch3:pfa_proof_2}
P_0\left[F_{\tau(c)+1} > r |\log c|\right] & \leq & P_0\left[\widehat{F}^*_{\tau(c)+1} > r |\log c|\right]. \nonumber
\end{eqnarray}
We have $P[\tau(c)+1>t]= P[\tau(c)>t-1] \leq k_1' \exp (-\eta t)$, where $k_1'=e^{\eta} E_0[e^{\eta \tau(c)}]$. Therefore, following \eqref{them_pe_borokov_thm},
\begin{eqnarray*}
\label{eq:ch3:proof_pfa_3rd_term_final}
\frac{P_0\left[F_{\tau(c)+1} > r |\log c|\right]}{c} &\leq & k_2' \frac{c^{r s_0(\eta)}}{c} \to 0,
\end{eqnarray*}
if $r s_0(\eta) > 1$ and $k_2'$ is a constant. We can choose $s_0(\eta) >1$ as in \eqref{them_pe_borokov_thm_use}. Then $\displaystyle \frac{1}{s_0(\eta)} < r \leq 1-\frac{1}{s'}$.\hfill \QED
\bibliographystyle{IEEEtranS}
\bibliography{IEEEabrv,mybib}

\begin{thebibliography}{10}
\providecommand{\url}[1]{#1}
\csname url@samestyle\endcsname
\providecommand{\newblock}{\relax}
\providecommand{\bibinfo}[2]{#2}
\providecommand{\BIBentrySTDinterwordspacing}{\spaceskip=0pt\relax}
\providecommand{\BIBentryALTinterwordstretchfactor}{4}
\providecommand{\BIBentryALTinterwordspacing}{\spaceskip=\fontdimen2\font plus
\BIBentryALTinterwordstretchfactor\fontdimen3\font minus
  \fontdimen4\font\relax}
\providecommand{\BIBforeignlanguage}[2]{{%
\expandafter\ifx\csname l@#1\endcsname\relax
\typeout{** WARNING: IEEEtranS.bst: No hyphenation pattern has been}%
\typeout{** loaded for the language `#1'. Using the pattern for}%
\typeout{** the default language instead.}%
\else
\language=\csname l@#1\endcsname
\fi
#2}}
\providecommand{\BIBdecl}{\relax}
\BIBdecl

\bibitem{Akyildiz_PC2011}
I.~Akyildiz, B.~Lo, and R.~Balakrishnan, ``Cooperative spectrum sensing in
  cognitive radio networks: A survey,'' \emph{Physical Communication}, vol.~4,
  no.~1, pp. 40--62, 2011.

\bibitem{Banerjee_WCOM}
T.~Banerjee, V.~Sharma, V.~Kavitha, and A.~K. JayaPrakasam, ``Generalized
  analysis of a distributed energy efficient algorithm for change detection,''
  \emph{{IEEE} Trans. Wireless Commun.}, vol.~10, no.~1, pp. 91 --101, Jan
  2011.

\bibitem{Barakat_SMA2004}
H.~M. Barakat and Y.~H. Abdelkader, ``Computing the moments of order statistics
  from nonidentical random variables,'' \emph{Statistical Methods \&
  Applications}, vol.~13, no.~1, pp. 15--26, 2004.

\bibitem{Billingsley_PM_Book}
P.~Billingsley, \emph{Probability and Measure}, 2nd~ed.\hskip 1em plus 0.5em
  minus 0.4em\relax John Wiley \& Sons, 1986.

\bibitem{Borovkov_1995}
A.~A. Borovkov, ``Unimprovable exponential bounds for distributions of sums of
  a random number of random variables,'' \emph{Theory of Probability \& Its
  Applications}, vol.~40, no.~2, pp. 230--237, 1995.

\bibitem{Dembo_LDP}
A.~Dembo and O.~Zeitouni, \emph{Large Deviations Techniques and Applications},
  2nd~ed.\hskip 1em plus 0.5em minus 0.4em\relax Springer.

\bibitem{Fellouris_TIT2011}
G.~Fellouris and G.~V. Moustakides, ``Decentralized sequential hypothesis
  testing using asynchronous communication,'' \emph{{IEEE} Trans. Inf. Theory},
  vol.~57, no.~1, pp. 534--548, 2011.

\bibitem{Govindarajulu_SS}
Z.~Govindarajulu, \emph{Sequential Statistics}.\hskip 1em plus 0.5em minus
  0.4em\relax World Scientific Pub Co Inc, 2004.

\bibitem{GUT_BOOK_2009}
A.~Gut, \emph{Stopped Random Walks: Limit Theorems and Applications},
  2nd~ed.\hskip 1em plus 0.5em minus 0.4em\relax Springer, 2009.

\bibitem{Iksanov_ECP}
A.~Iksanov and M.~Meiners, ``Exponential moments of first passage times and
  related quantities for random walks,'' \emph{Electronic Communications in
  Probability}, vol.~15, pp. 365--375, 2010.

\bibitem{Janson_AAP}
S.~Janson, ``Moments for first-passage and last-exit times, the minimum, and
  related quantities for random walks with positive drift,'' \emph{Advances in
  applied probability}, pp. 865--879, 1986.

\bibitem{Lai_AS1988}
T.~L. Lai, ``Nearly optimal sequential tests of composite hypotheses,''
  \emph{The Annals of Statistics}, pp. 856--886, 1988.

\bibitem{Li_WCOM2010}
H.~Li, H.~Dai, and C.~Li, ``Collaborative quickest spectrum sensing via random
  broadcast in cognitive radio systems,'' \emph{{IEEE} Trans. Wireless
  Commun.}, vol.~9, no.~7, pp. 2338--2348, Jul 2010.

\bibitem{Mei_TIT2008}
Y.~Mei, ``Asymptotic optimality theory for decentralized sequential hypothesis
  testing in sensor networks,'' \emph{{IEEE} Trans. Inf. Theory}, vol.~54,
  no.~5, pp. 2072 --2089, May 2008.

\bibitem{Page_Biometrika1954}
E.~Page, ``Continuous inspection schemes,'' \emph{Biometrika}, vol.~41, no.
  1/2, pp. 100--115, 1954.

\bibitem{POO_BOOK_1}
H.~V. Poor and O.~Hadjiliadis, \emph{Quickest Detection}, 1st~ed.\hskip 1em
  plus 0.5em minus 0.4em\relax Cambridge University Press, 2008.

\bibitem{Quan_SPM2008}
Z.~Quan, S.~Cui, H.~Poor, and A.~Sayed, ``Collaborative wideband sensing for
  cognitive radios,'' \emph{{IEEE} Signal Process. Mag.}, vol.~25, no.~6, pp.
  60 --73, Nov 2008.

\bibitem{Rootzen_AAP1988}
H.~Rootz{\'e}n, ``Maxima and exceedances of stationary markov chains,''
  \emph{Advances in applied probability}, pp. 371--390, 1988.

\bibitem{Ross_SP_Book}
S.~M. Ross, \emph{Stochastic Processes}, 2nd~ed.\hskip 1em plus 0.5em minus
  0.4em\relax Wiley, 1995.

\bibitem{Saaty_book}
T.~L. Saaty, \emph{Nonlinear Integral Equations}.\hskip 1em plus 0.5em minus
  0.4em\relax Dover Publications, 1981.

\bibitem{Shei_PIMRC2008}
Y.~Shei and Y.~T. Su, ``A sequential test based cooperative spectrum sensing
  scheme for cognitive radios,'' in \emph{Proc. IEEE 19th International
  Symposium on Personal, Indoor and Mobile Radio Communications (PIMRC)}, Sep
  2008.

\bibitem{Siegmund_SATC_Book}
D.~Siegmund, \emph{Sequential analysis: tests and confidence intervals}.\hskip
  1em plus 0.5em minus 0.4em\relax Springer, 1985.

\bibitem{Unnikrishnan_JSTSP2008}
J.~Unnikrishnan and V.~V. Veeravalli, ``Cooperative sensing for primary
  detection in cognitive radio,'' \emph{{IEEE} J. Sel. Topics Signal Process.},
  vol.~2, no.~1, pp. 18 --27, Feb 2008.

\bibitem{Veeravalli1999}
V.~V. Veeravalli, ``Sequential decision fusion: theory and applications,''
  \emph{Journal of the Franklin Institute}, vol. 336, no.~2, pp. 301--322,
  1999.

\bibitem{Yilmaz_Asilomar2012}
Y.~Yilmaz, G.~Moustakides, and X.~Wang, ``Spectrum sensing via event-triggered
  sampling,'' in \emph{Proc. Forty Fifth Asilomar Conference on Signals,
  Systems and Computers (ASILOMAR),}, Nov, pp. 1420--1424.

\bibitem{Yilmaz_Allerton2012_journal}
Y.~Yilmaz, G.~V. Moustakides, and X.~Wang, ``Channel-aware decentralized
  detection via level-triggered sampling,'' \emph{{IEEE} Trans. Signal
  Process.}, vol.~61, no.~2, pp. 300--315, 2013.

\bibitem{Zou_TSP}
Q.~Zou, S.~Zheng, and A.~H. Sayed, ``Cooperative sensing via sequential
  detection,'' \emph{{IEEE} Trans. Signal Process.}, vol.~58, no.~12, pp. 6266
  --6283, Dec 2010.

\end{thebibliography}
\end{document}